\documentclass[12pt,a4paper]{article}
\usepackage[toc,page]{appendix}

\usepackage{geometry}
\usepackage{epsfig}
\usepackage{float}
\usepackage{amsfonts}
\usepackage{amsopn}
\usepackage{amssymb}
\usepackage{amsmath}
\usepackage{graphicx}
\usepackage{amsthm}
\usepackage{mathtools}
\usepackage{esint}
\usepackage[T1]{fontenc}
\usepackage[utf8]{inputenc}
\usepackage{authblk}
\usepackage{wasysym}
\usepackage{tikz}
\usepackage{verbatim}
\usetikzlibrary{shapes,snakes}
\numberwithin{equation}{section}
\newtheorem{lemma}{Lemma}
\newtheorem{theorem}{Theorem}

\date{}

\title{Optimized Fock space in the large $N$ limit of quartic interactions in Matrix Models}

\author{ Mariusz Hynek\footnote{mkhynek@kth.se}}
\affil{Department of Mathematics,\\
KTH Royal Institute of Technology,\\
100 44 Stockholm,
Sweden}

\begin{document}
\maketitle

\abstract{We consider the problem of quantization of the bosonic membrane via the large $N$ limit of its matrix regularizations $H_N$ in Fock space. We prove that there exists a choice of the Fock space frequency such that $ H_N$ can be written as a sum of a non-interacting Hamiltonian $H_{0,N}$ and the original normal ordered quartic potential. Using this decomposition we obtain upper and lower bounds for the ground state energy in the planar limit, we study a perturbative expansion about the spectrum of $H_{0,N}$, and show that the spectral gap remains finite at $N=\infty$ at least up to the second order. We also apply the method to the $U(N)$-invariant anharmonic oscillator, and demonstrate that our bounds agree with the exact result of Brezin et al.}

\newpage
\tableofcontents

\section{Introduction}
Let us consider the classical internal energy of the bosonic membrane, which in a light-cone description in orthonormal gauge can be written as (for more details see e.g. \cite{relmemb})
\begin{align}
\mathbb{M}^2&= \int_{\Sigma} \left( \frac{\vec{p}^2}{\rho}+\rho\sum_{i<j}\lbrace x_i, x_j\rbrace\right) d^{2}\varphi,
\end{align}
with the constraints 
\begin{equation}
\sum_{i=1}^d \lbrace x_i,p_i\rbrace=0,\label{contraints continuum}
\end{equation}
where the integral is performed over a 2-dimensional compact manifold $\Sigma$ and $\lbrace f,g\rbrace:=\frac{1}{\rho(\varphi)}(\partial_1 f \partial_2g-\partial_2 f \partial_1 g)$ denotes the Poisson bracket. It is convenient to use the mode expansions $x_i(\varphi)= x_{i \alpha} Y_{\alpha}(\varphi),~~ p_j(\varphi)= p_{j \alpha} Y_{\alpha}(\varphi)$ in terms of the eigenfunctions $\lbrace Y_{\alpha}\rbrace_{\alpha=1}^{\infty}$ of the Laplace operator on $\Sigma$, where the zero modes are subtracted. This allows to rewrite $\mathbb{M}^2$ as an infinite sum over the internal modes
\begin{align}
\mathbb{M}^2&=p_{i \alpha}p_{i \alpha} +\frac{1}{2}g_{\alpha \beta \gamma} g_{\alpha \beta' \gamma'}x_{i\beta} x_{i\beta'}x_{j\gamma} x_{j\gamma'},\label{field theory}\\
g_{\alpha \beta \gamma}&:=\int{Y_{\beta}\epsilon^{a b}\partial_{a}Y_{\alpha}\partial_{b}Y_{\gamma}d^{2}\varphi},~~i=1,...,d, ~~~\alpha, \beta, \gamma=1,...,\infty.
\end{align}

It has been shown by Goldstone and Hoppe  \cite{Hoppe_phd, hoppe02} that the full field-theoretic Hamiltonian (\ref{field theory}) admits a regularization procedure, where the classical phase-space variables $x_i(\varphi), p_j(\varphi)$ are replaced by $n$-dimensional matrices, the Poisson bracket by the matrix commutator and integrals over $\Sigma$ by the matrix trace. The original volume-preserving diffeomorphisms symmetry of $\Sigma$, represented by (\ref{contraints continuum}), is recovered in the $n \rightarrow \infty$ limit from the $SU(n)$ invariance of its matrix regularizations. The family of $n$-dimensional matrix models constructed in this way reads 
\begin{equation}
H_N= Tr(\vec{P}^2)-(2 \pi n)^2 n\sum_{i<j}^d Tr([X_i,X_j]^2),
\label{matrix model}
\end{equation}
with the $SU(n)$ invariance constraints
\begin{equation}
 \sum_{i=1}^d  [X_i,P_i]=0,
\end{equation}
where $P_i, X_i$ are hermitian traceless $n \times n$ matrices. The scaling factor in front of the quartic potential is chosen in such a way that $\lim_{N \rightarrow \infty} H_N=\mathbb{M}^2$.

Using a basis of $su(n)$, $T_a$, $a=1,...,n^2-1:=N$, with $Tr(T_a T_b)=  \delta_{ab}$ and $[T_a,T_b]=i \hbar_n \frac{1}{\sqrt{n}} f^{(n)}_{a b c} T_c$,  $\hbar_n = \frac{1}{2 \pi n}$, $f^{(n)}_{abc}=\frac{2 \pi n^{\frac{3}{2}}}{i}Tr(T_a \left[T_b,T_c \right])$, we can rewrite $H_N$ (and the constraints) in terms of $d(n^2-1)$ canonical pairs $p_{ia}, x_{ia}$ ($X_i=x_{ia}T_a, P_i=p_{ia}T_a$) as a \emph{finite} sum over the matrix modes, cp.(\ref{field theory}),
\begin{align}
H_N(p,x) &= p_{ia} p_{ia} +\frac{1}{2} f^{(n)}_{abc}f^{(n)}_{ab'c'}x_{ib}x_{ib'}x_{jc}x_{jc'},
\label{classmembr}
\\ f_{abc}^{(n)}x_{ib}p_{jc}&=0.
\label{constraints}
\end{align}

In contrast to string theory, the Hamiltonian of the membrane with its quartic interaction makes the problem of quantisation rather difficult. One approach, which was proposed in the literature many years ago \cite{Hoppe_phd, hoppe02}, is to take advantage of the symmetry preserving matrix regularizations (\ref{matrix model})/(\ref{classmembr}), quantize it for finite $N$ and then take the limit $N \rightarrow \infty$. While it has been proved that all finite $N$ expressions are well defined Schrödinger operators \footnote{See \cite{douglas thesis} for a discussion of the spectrum, including the supersymmetric version of the model and related issues} with purely discrete spectrum \cite{simon,Luscher}, it seems that almost nothing is known about the large $N$ limit (apart from the case $d=1$, where the quartic potential vanishes \cite{33}). In particular, it is interesting to ask whether the spectrum remains purely discrete in the limit and how it scales with $N$ and $d$. In this paper we study the large $N$ behaviour of the canonically quantized family of Hamiltonians (\ref{classmembr}) based on a Fock space description in order to approach those questions. \\
According to common knowledge, in order to quantize a system with many degrees of freedom one should rescale the quartic interaction by the suitable power\footnote{The power of $n$ for a quartic interaction it is usually argued to be $-1$, (the t' Hooft coupling \cite{tHooft}), however note that due to the definition of $f_{abc}^{(n)}$ which contains an explicit factor of $n^{\frac{3}{2}}$, the coupling constant in front of our potential should be multiplied by $n^{-4}$. In Section 4 we will show that this also follows from our construction} of $n$ to make the quadratic part competitive at large $n$. This can be realized as a rescaling of the classical phase-space variables preserving the canonical Poisson-commutation relations (and the form of the constraints), i.e. $x_{ia} \rightarrow N^{-\alpha} x_{ia}, p_{ia} \rightarrow N^{\alpha} p_{ia}$, which leads to the following classical energy
\begin{equation}
H_N(p,x)=p_{ia} p_{ia} +N^{-\gamma}\frac{1}{2} f^{(n)}_{abc}f^{(n)}_{ab'c'}x_{ib}x_{ib'}x_{jc}x_{jc'},
\label{generalresc}
\end{equation}
where by abuse of notation we denote the rescaled operator $N^{-2 \alpha} H_N$ by the same symbol $H_N$ and $\gamma:=6 \alpha$.
The canonically quantized expression corresponding to (\ref{generalresc} ) is a priori ill-defined in the large $n$ limit due to possible divergences coming from the infinite vacuum energy and therefore it needs to be renormalized by subtracting a multiple of the identity operator. This is a rather common property of quantum models with many degrees of freedom, so in order to see the general pattern let us consider a whole class of Hamiltonians parametrized by a sequence of real tensors $c^{(N)}_{IJKL}$, $I,J,K,L=1,...,N$:
\begin{equation}
H_N(p,x)=\sum_{I \in \mathcal{J}_N} p_I p_I  + \sum_{I \in \mathcal{J}_N} \omega_{0I}^2 x_I x_I + N^{- \gamma} \sum_{I,J,K,L \in \mathcal{J}_N} c^{(N)}_{IJKL}x_I x_J x_K x_L,
\label{general}
\end{equation}
where $ \mathcal{J}_N$ is the index set having the form of a cartesian product of two discrete sets $ \mathcal{J}_N= \mathcal{D} \times \mathcal{K}_n$, with $|\mathcal{D}|=d=const$, $|\mathcal{K}_n|=N  \nearrow \infty$. For instance, for the Membrane Matrix Models (MMM) (\ref{classmembr}) we have $N=n^2-1$, $\mathcal{K}_n= \lbrace1,...,n^2-1 \rbrace $, $\mathcal{D}=\lbrace 1,...,d\rbrace $, $\omega_{0I}=0$ and $c_{IJKL}\equiv c_{(bi)(b'i')(cj)(c'j')}=\frac{1}{2} f_{abc}f_{ab'c'}\delta_{ii'}\delta_{jj'}$.\\

It is known that a Fock space description of local Hamiltonians of the form $p^2+q^2+V(x)$ is rather inconvenient due to the occurence of terms containing only annihilation or only creation operators, which implies that the so-called single trace sector is not invariant under the action of the Hamiltonian and the true vacuum is not a simple gaussian in the variables $x_I$ (see \cite{Fock N} for a review). Therefore the problem of finding eigenvalues becomes in most cases extremely difficult due to the rapidly increasing with $N$ number of relevant degrees of freedom. We will however show that despite all of that, one can get certain spectral properties of local Hamiltonians (\ref{general}) using Fock space methods, such as upper and lower bounds for the spectrum in the planar limit, as well as a qualitative picture of the subtle interplay between the quadratic terms and the interaction leading to a redefinition of the vacuum energy and mass. The tools described in this paper are of high relevance especially for multi matrix-models where other methods based on the diagonalisation of the matrix degrees of freedom and integrating out the "angle" variables are not directly applicable, like for the Membrane Matrix Models (MMM).

The rest of the paper is organized as follows.
In the next section we will introduce the notion of optimized Fock space providing a very convenient decomposition of the quantum Hamiltonian corresponding to (\ref{general}) and giving directly the gaussian variational bound for the ground state energy. The idea consists in tuning the frequencies of the harmonic oscillators, whose eigenfunctions are used as a basis of the corresponding Fock space. We give an algorithm how to choose an optimal set of these frequencies such that the form of the Hamiltonian (\ref{general}) expressed in the language of creation and annihilation operators is the simplest possible. This almost trivial observation (usually not spelled out explicitly in the literature though, see however \cite{beyond}), based on the fact that the optimal Fock space frequencies in the quantum representation of (\ref{general}) are not always the $\omega_{0I}$'s, is the starting point of our study of the MMM (\ref{generalresc}), a system where there is no obvious choice of a basis consisting of harmonic oscillators since all the $\omega_{0I}$'s are equal to zero. Thus one has to explore the quartic potential in more detail, which gives birth to a mass term (quadratic in $x_I$) in a properly chosen Fock space. In Section 3 we consider a toy model, the $U(n)$-invariant anharmonic oscillator (AO), and present Fock space-based techniques to get upper and lower bounds for the ground state energy in the planar limit, which agree with the exact answer \cite{planar,collective,singlet spectrum mondello onofri,shapiro, Yaffe, planar limit marchesini onofri}. In Section 4 we apply the method to the MMM. Finally, in Section 5 we show that the perturbative expansion suggested by the optimized Fock space decomposition serves as a very good approximation for the vacuum energy and for the spectral gap of the AO at least up to the third order, even in the strong-coupling regime $g \rightarrow \infty $, restoring the correct scaling of the spectrum with the coupling constant $\propto g^{\frac{1}{3}}$. We also observe that the MMM contains a small hidden parameter (an effective coupling constant) and we perform the corresponding expansion, which turns out to be consistent with our lower and upper bounds for the vacuum energy, justifying the validity of the planar limit for this model. In particular, we find that the effective coupling constant is proportional to $\frac{1}{d-1}$, hence the perturbative expansion has better convergence properties in higher dimensions. Moreover, the perturbation series for the first $SO(d)\times SU(n)$ invariant excited state indicates that the spectral gap is finite at large $n$.

Although the purely bosonic model (\ref{quantum membrane}) is interesting on its own, more attention has been recently paid to its supersymmetric extensions, in particular for $d=9$ dimensions (resp. for $d=9$ matrices), see e.g. \cite{supermembrane} and \cite{douglas thesis} for a more recent review of this topic. Remarkably, the spectrum of the supersymmetric version of (\ref{quantum membrane}) turns out to be continuous and equal to the interval $[0, \infty)$, \cite{unstable supermembrane}. However, there have been also evidence of existence of discrete eigenvalues embedded in the continuous spectrum  \cite{akm, coexisting wosiek}. In particular, zero energy normalizable states are of high importance for the supermembrane as well as in the context of the BFSS conjecture of M-Theory \cite{BFSS}. Despite a number of profound results concerning zero energy eigenfunctions, e.g. an explicit construction in the pure fermionic sector for the $n=2$ model \cite{wosiek}, large $x$ behaviour \cite{largex}, or existence results on compact regions \cite{compact}, the question posed in the full generality remains open. We believe that the approach discussed in this paper can shed new light on this problem since our results correspond to the purely bosonic sector of the supermembrane and our construction can be possibly extended to the fermionic sectors allowing to study the embedded part of the spectrum of the supermembrane. Moreover, as our method simplifies the Fock space representation of the Hamiltonian, it should also allow to optimize cut-off Fock space algorithms for studying various quantum mechanical systems related to (Super)Yang-Mills theories, introduced in \cite{cut-off}.

\section{Optimized Fock space}
\label{opt Fock space}

The canonically quantized Hamiltonian (\ref{general}) becomes formally a Schrödinger operator acting on $\otimes_{I=1}^{N} L^2(\mathbb{R},dx^I)$ with the classical coordinates replaced by operators in the usual way, i.e. $p_I=-i \partial_I$ and $x_I$ being the multiplication operator. Since we are interested in the large $N$ limit, it is convenient to embed $\otimes_{I=1}^{N} L^2(\mathbb{R},dx^I)$ in the standard bosonic Fock space $\mathcal{H}_{\omega}$, defined as the Hilbert space generated by states of the form
\begin{equation}
\psi_{\lbrace I_1,...,I_k \rbrace}:= a_{I_1}^{\dagger}(\omega_{I_1}) ...a_{I_k}^{\dagger}(\omega_{I_k})\Psi_0(\omega),~~k ~\text{finite},
\label{elmentary}
\end{equation}where 
\begin{align}
a_I(\omega_I)=\frac{1}{\sqrt{2}}\left(\frac{\partial_I}{\sqrt{\omega_I}}+ \sqrt{\omega_I} x_I \right), \nonumber \\
a_I^{\dagger}(\omega_I)=\frac{1}{\sqrt{2}}\left(-\frac{\partial_I}{\sqrt{\omega_I}} + \sqrt{\omega_I} x_J \right),
\end{align}
with
\begin{equation}
[a_I(\omega),a_J^{\dagger}(\omega)]= \delta_{I J} \mathbb{I}.
\label{alg}
\end{equation}
We denote the norm and scalar product in  $\mathcal{H}_{\omega}$ by $||.||$ and $\langle.,.\rangle $ respectively. The vacuum state $\Psi_0(\omega)$ has the form of an infinite product $\Psi_0(\omega):= \Pi_{I=1}^{\infty} \psi_{\omega_I}(x_I)  $ with $\psi_{\omega_I}(x_I):= \sqrt[4]{\frac{ \omega }{\pi}} e^{-\frac{1}{2} \omega_I x_I^2}, ~a_I(\omega_I)\Psi_0(\omega)=0~ \forall I $.          \\
Then
\begin{align}
p_I(\omega_I)=-\frac{i \sqrt{\omega_I}}{\sqrt{2}}(a_I- a_I^{\dagger}), \nonumber \\
x_I(\omega_I)=\frac{1}{\sqrt{2 \omega_I}}(a_I + a_I^{\dagger}).
\label{pandx}
\end{align}
Since $\otimes_{I=1}^{N} L^2(\mathbb{R},dx^I) $ is isomorphic to the subspace $\mathcal{H}_{\omega,N} \subset \mathcal{H}_{\omega}$ generated by the first $N$ creation operators $a_I^{\dagger}, I=1,...,N$, the action of $H_N$ can be naturally extended to the whole Fock space by taking the tensor product of $H_N$ with the identity acting on $\mathcal{H}_{\omega,N}^{\perp}$. Therefore the operator induced on $\mathcal{H }_{\omega}$ by $H_N$ has the form
\begin{equation}
H_N=\left(\sum_I p_I p_I +  \sum_{I}\omega_{0I}^2 x_I x_I+ N^{-\gamma} \sum_{I,J,K,L}c^{(N)}_{IJKL}x_I x_J x_K x_L \right ) \otimes \mathbb{I}_{\mathcal{H}_{\omega,N}^{\perp}}.
\label{fullquantfN}
\end{equation}
$H_N$ is well defined for finite N (if  $c^{(N)}_{IJKL}$ are finite), but in general not for $N=\infty$. In order to assure that the domain of $H_N$ contains more than the zero vector, one has to subtract the divergent ground state energy by adding a multiple of the identity operator $\beta_N \mathbb{I}$ to the Hamiltonian. We expect that for $\gamma$ large enough (and $ c^{(N)}_{IJKL}$ not too pathological) one should get at least a finite limit of $ ||(H_{N}+ \beta_N \mathbb{I})\psi||$ for a generic $\psi\in \mathcal{H}$ and some $\beta_N$. This is unfortunately not always the case since the kinetic energy contains terms proportional to $a_I a_I$ and $a_I^{\dagger} a_I^{\dagger}$, which are divergent at large $N$ because e.g. $||\sum_I a_I^{\dagger} a_I^{\dagger} \psi||\rightarrow \infty$ for every non-zero $\psi \in \mathcal{H}_{\omega}$. As a consequence, in order to simplify the large $N$ behaviour of the theory one should cancel such terms by finding corresponding counter-terms in the potential. We will see that this is possible for a relatively large class of models including all models with a large symmetry group, e.g. $SO(n)$ symmetric vector models or $U(n)/SU(n)$ symmetric matrix models, where the interaction has the form of a trace. The idea is that the sequence of the Fock space frequencies $\lbrace \omega_I \rbrace$ has to be adjusted to the Hamiltonian, i.e. the choice of the subspace $\mathcal{H}_{\omega} \subset \otimes_{I=1}^{\infty}L_2(\mathbb{R},dx^I)$ where the limit is taken, depends on the operator and in most cases the optimal choice is not the natural choice $\omega_I= \omega_{0I}$ (especially when all $\omega_{0I}=0$ like for the MMM). 

Let us see how it works in detail. $H_N$ rewritten in terms of $a_I$ and $a_I^{\dagger}$ becomes (for the sake of transparency, we leave out the $\mathbb{I}_{\mathcal{H}_{\omega,N}^{\perp}}$ part of (\ref{fullquantfN}))
\begin{align}
H_N \equiv T_N+V_N^{(2)}+V_N^{(4)} = \frac{1}{2}\sum_I \omega_I(2a_I^{\dagger}a_I-a_I^{\dagger}a_I^{\dagger}-a_I a_I +  \mathbb{I})+\frac{1}{2}\sum_I \frac{\omega_{0I}^2}{\omega_I}(2a_I^{\dagger}a_I+a_I^{\dagger}a_I^{\dagger}+a_I a_I +  \mathbb{I})\nonumber \\+\frac{ N^{-\gamma}}{4 }\sum_{IJKL}\frac{c_{IJKL}}{\sqrt{\omega_I\omega_J \omega_K \omega_L}}(a_I a_J a_K a_L +a_I^{\dagger} a_J a_K a_L+a_I a_J^{\dagger} a_K a_L+a_I a_J a_K^{\dagger} a_L+a_I a_J a_K a_L^{\dagger}\nonumber \\
+a_I^{\dagger} a_J^{\dagger} a_K a_L+a_I^{\dagger} a_J a_K^{\dagger} a_L+a_I^{\dagger} a_J a_K a_L^{\dagger}+a_I a_J^{\dagger} a_K^{\dagger} a_L+a_I a_J^{\dagger} a_K a_L^{\dagger}+a_I a_J a_K^{\dagger} a_L^{\dagger}\nonumber \\
a_I^{\dagger}a_J^{\dagger}a_K^{\dagger}a_L+a_I^{\dagger}a_J^{\dagger}a_K a_L^{\dagger}
+a_I^{\dagger}a_Ja_K^{\dagger}a_L^{\dagger}+a_I a_J^{\dagger}a_K^{\dagger}a_L^{\dagger}
\nonumber \\ +a_I^{\dagger}a_J^{\dagger}a_K^{\dagger}a_L^{\dagger}).
\end{align}
We rewrite the quartic potential using the commutation relations (\ref{alg})
\begin{align}
V_N^{(4)}= N^{-\gamma}\frac{1}{4 } \sum_{IJK}\frac{1}{\omega_K\sqrt{\omega_I\omega_J }}\left(\frac{1}{2} a_I^{\dagger} a_Jc_{(I J KK)}+
\frac{1}{4}(a_I a_J+a_I^{\dagger} a_J^{\dagger})c_{(I J KK)}\right)\\
 +\frac{ N^{-\gamma}}{4} \sum_{I,J}\frac{1}{\omega_I \omega_J}(c_{IIJJ}+c_{IJIJ}+c_{IJJI}) \mathbb{I} + N^{-\gamma}:V_N: \\
 \equiv N^{-\gamma} A^{(N)}_{IJ} a_I^{\dagger} a_J+
 N^{-\gamma}\frac{1}{2} A^{(N)}_{IJ} (a_I a_J + a_I^{\dagger} a_J^{\dagger})+\\
 +N^{-\gamma}f(N)\mathbb{I}+ N^{-\gamma}:V_N:,
\end{align}
where we have defined $A^{(N)}_{IJ}:=\sum_{K}\frac{c^{(N)}_{(I J KK)}}{8 \omega_K \sqrt{\omega_I\omega_J}}$ and $f(N):=\sum_{I,J}\frac{c^{(N)}_{IIJJ}+c^{(N)}_{IJIJ}+c^{(N)}_{IJJI}}{4 \omega_I \omega_J} $, $\lbrace (I_1,...,I_k) \rbrace :=\sum_{\pi \in S_k}\lbrace I_{\pi(1)},...,I_{\pi(k)} \rbrace$ denotes the symmetrization and $::$ is the normal ordering with respect to $\Psi_0({\omega})$. We get 
 
\begin{align}
H_N+\beta_N \mathbb{I}=\sum_{I,J} \lbrace ((\omega_I+\frac{\omega_{0I}^2}{\omega_I}) \delta_{IJ} + N^{-\gamma} A^{(N)}_{IJ})a_I^{\dagger}a_J +
\frac{1}{2}(N^{-\gamma} A^{(N)}_{IJ} -(\omega_I+\frac{\omega_{0I}^2}{\omega_I}) \delta_{IJ})(a_I a_J + a_I^{\dagger} a_J^{\dagger}) \rbrace\\+  \left(\frac{1}{2}\sum_I (\omega_I+\frac{\omega_{0I}^2}{\omega_I})+N^{-\gamma}f(N)+\beta_N \right)\mathbb{I}+ N^{-\gamma}:V_N: \\
 \equiv \sum_{I,J} \left( A^{(N+)}_{IJ} a_I^{\dagger}a_J +  \frac{1}{2} A^{(N-)}_{IJ}(a_I a_J + a_I^{\dagger} a_J^{\dagger})\right)+ N^{-\gamma}:V_N:,
\end{align}
where we have chosen $\beta_N:= -\frac{1}{2}\sum_I(\omega_I+\frac{\omega_{0I}^2}{\omega_I})-N^{-\gamma}f(N)$ and have introduced two more matrices $A^{(N+)}_{IJ}:=N^{-\gamma} A^{(N)}_{IJ}+(\omega_I+\frac{\omega_{0I}^2}{\omega_I}) \delta_{IJ}$ and $A^{(N-)}_{IJ}:=N^{-\gamma} A^{(N)}_{IJ} -(\omega_I-\frac{\omega_{0I}^2}{\omega_I}) \delta_{IJ}$.\\
Let us now assume that there exists a sequence $\lbrace \tilde{\omega}_I \rbrace$, called the optimized Fock space frequencies (and $\mathcal{H}_{\tilde{\omega}}$ called the optimized Fock space), s.t.  
\begin{equation}
 \lim_{N \rightarrow \infty} A^{(N-)}_{IJ}=0,~~\forall I,J,
 \label{diagonallimit}
\end{equation}
i.e. that the matrix $A^{(N)}_{IJ}$ is diagonal at large $N$ (this is in fact the case for the mentioned above $SO(n)$ symmetric vector models or $U(n)/SU(n)$ symmetric matrix models\footnote{One way to see it is to notice that for these models all double contractions of the tensor defining the quartic interaction produce Kronecker deltas, i.e. $c_{IJKK}^{(N)}\propto \delta_{IJ}$} ). The smallest $\gamma$ for which it is possible we call $\gamma_{crit}$.
 From eq. (\ref{diagonallimit}) we get that $  \lim_{N \rightarrow \infty}  N^{- \gamma}A^{(N)}_{IJ}=diag((\tilde{\omega}_1-\frac{\omega_{0I}^2}{\tilde{\omega}_1}),(\tilde{\omega}_2-\frac{\omega_{0I}^2}{\tilde{\omega}_2}),...)$ and thus
   $ \lim_{N \rightarrow \infty} A^{(N+)}_{IJ}=diag(2 \tilde{\omega}_1, 2\tilde{\omega}_2,...)$.
Moreover the sequence $\beta_N$ exhibits a very nice property given in the following lemma. 
 \begin{lemma}
 \label{lemma1}
({\bf Optimized Fock space decomposition}) Assume that $ A_{IJ}^{(N-)}\simeq O(\frac{1}{N})$. Then the optimized Fock space frequencies $\tilde{\omega}_I$ coincide with the optimized gaussian vacuum frequencies at large $n$ and $-\frac{\beta_N}{N}$ converges to the upper gaussian variational bound $e_{0,N}^{(0)}$ for the ground state energy of $\frac{H_N}{N}$ as $n \rightarrow \infty$. Moreover, the Hamiltonian (\ref{fullquantfN}) admits the following decomposition in $\mathcal{H}_{\tilde{\omega}}$
\begin{equation}
H_N=\left( 2\sum_{I=1}^N \tilde{\omega}_I a_I^{\dagger}a_I+\frac{1}{N^{- \gamma}}:V_{N}:+ N e_{0}^{(0)} \mathbb{I} \right) \otimes  \mathbb{I}_{\mathcal{H}_{\tilde{\omega},N}^{\perp}} + R_N,
\end{equation}
where $e_{0}^{(0)}= -\lim_{n \rightarrow \infty} \frac{\beta_N}{N} =\lim_{n \rightarrow \infty} e_{0,N}^{(0)}$ is given by the condition $\lim_{n \rightarrow \infty} A_{IJ}^{(N-)}=0$, and \\$\lim_{n \rightarrow \infty} ||R_N \psi||=0,~ \forall \psi \in \mathcal{H}_{\tilde{\omega}}$.
 \end{lemma}

  \begin{proof}
   By noting that $\langle \Psi_0(\omega), H_N \Psi_0(\omega)\rangle=-\beta_N$ and using the variational principle we get that the optimized gaussian frequencies satisfying 
 \begin{align}
 0=\frac{\partial \beta_N}{\partial \omega_I}=\frac{1}{2}(1-\frac{\omega_{0I}^2}{\omega_I^2})-\frac{N^{- \gamma}}{4 \omega_I^2}\sum_J\frac{1}{\omega_J}(c^{(N)}_{IIJJ}+c^{(N)}_{JJII}+c^{(N)}_{IJJI}+c^{(N)}_{JIIJ}+c^{(N)}_{IJIJ}+c^{(N)}_{JIJI}),
 \end{align}
 or equivalently 
 \begin{align}
 \omega_I^2=\omega_{0I}^2 +\frac{N^{-\gamma}}{8}\sum_J\frac{1}{\omega_J}c^{(N)}_{(IIJJ)},
 \label{renormalized frequencies}
 \end{align}
converge to the solutions of (\ref{diagonallimit}). Moreover, the rest term $R_N$ originates from the matrix $ A_{IJ}^{(N-)}$, which means that $R_N \simeq \frac{const.}{N}\sum_I (a_I a_I + a_I^{\dagger} a_I^{\dagger} )$ and therefore $ ||R_N \psi||_{\tilde{\omega} } \rightarrow 0 ~\forall \psi \in \mathcal{H}_{\tilde{\omega}}$.
\end{proof} 

Note that this result does not imply that the system described by the Hamiltonian (\ref{general}) becomes necessarily a system of decoupled harmonic oscillators despite that according to Lemma \ref{lemma1}
\begin{equation}
\lim_{N \rightarrow \infty} \langle \psi, (H_N+\beta_N \mathbb{I}) \phi \rangle=\langle \psi,2 \sum_{I} \tilde{\omega}_I a_I^{\dagger}a_I \phi \rangle,
\end{equation}
$ \forall  \phi, \psi \in \mathcal{H}_{\tilde{\omega}}$ containing a finite number of elementary excitations (\ref{elmentary}). This means that, in order to avoid such a trivialisation, one has to employ states with infinitely many oscillatory modes and treat the optimized Fock space as the first step towards the correct quantum description of the model. This is the case for the two considered here matrix models, as we will see in the next sections, where the relevant space of interest is the space of $SU(n)$ invariants.

\section{$U(n)$-invariant anharmonic oscillator}

Let us consider the $U(n)$ symmetric matrix model 
\begin{align}
2 H_N=Tr(P^2)+Tr( M^2+ \frac{2 g}{n} M^4),
\label{1matrix}
\end{align}
where $M$ is a hermitian $n \times n$ matrix and $P$ its conjugate momentum.
Since the exact value of the ground state energy of the model is known and even the $U(n)$ symmetric sector has been solved in the large $n$ limit one could expect that the Fock space approach should allow to rederive these results. Unfortunately, due to various technical difficulties, such as the fact that the Fock space representation of the Hamiltonian (\ref{1matrix}) does not annihilate the Fock space vacuum and does not preserve the single trace sector, there has been no exact solutions based on a Fock space formalism for this model so far. Nevertheless, one can still get some information about the spectrum using Fock space methods, which is of high importance for models where other methods fail. The Hamiltonian (\ref{1matrix}) is an excellent laboratory for testing our large $N$ techniques before approaching the Membrane Matrix Models (\ref{matrix model}), since it exhibits all the properties which make the quantization of a system with a large number of degrees of freedom cumbersome.\\  Let us start with the optimized Fock space decomposition for (\ref{1matrix}). 
  By expanding $M$ and $P$ in the basis $\lbrace T_a \rbrace_{a=1}^{n^2}$ consisting of $N:=n^2$ generators of $U(n)$ with normalisation $Tr(T_a T_b)=\delta_{ab}$, satisfying the completeness relation
\begin{equation}
 (T_a)_{ij}(T_a)_{kl}=\delta_{jk} \delta_{il},
 \label{completeness relation1}
\end{equation} 

we arrive at   $M=T_a x_a, P=T_a p_a$ ($N=n^2-1,~\mathcal{K}_N = \lbrace 1,...,n^2-1 \rbrace, ~ d=1,~ \omega_{0I}=1$) and
\begin{align}
2 H_N= p_a p_a + x_a x_a + c_{abcd} x_a x_b x_c x_d,
\label{q1matrix}
\end{align}
with $c_{abcd}=\frac{2 g}{n} Tr(T_a T_b T_c T_d)$. The action of $H_N$ in the Fock space $\mathcal{H}_{\omega}$, given in terms of the creation and annihilation operators, becomes (assuming that $\omega_a= \omega~ \forall a$)
\begin{align}
 2H_N=  \sum_{a,b} \left( A^{(N+)}_{ab} a_a^{\dagger}a_b +  \frac{1}{2} A^{(N-)}_{ab}(a_a a_b + a_a^{\dagger} a_b^{\dagger})\right)+ \frac{2 g}{n} :Tr{M^4}:- \beta_N \mathbb{I}.
\end{align}
As mentioned previously, the matrix $A_{ab}^{(N)}$ is diagonal\footnote{$A_{IJ}^{(N)}\simeq O(1)$ and thus $\gamma_{crit}=0$. Note also that we perform the calculation for $2 H_N$ instead of $H_N$ in order to match it with the conventions from Section 2 and in the end we divide the results by 2 to compare them with the work of Brezin et al.} and 
\begin{align}
A^{(N+)}_{ab}&=(\omega +\frac{1}{\omega}+\frac{4g}{\omega^2})\delta_{ab}+\frac{2g}{n} \delta_{a0} \delta_{b0},\\
A^{(N-)}_{ab}&=(-\omega +\frac{1}{\omega}+\frac{4g}{\omega^2})\delta_{ab}+\frac{2g}{n} \delta_{a0} \delta_{b0},\\
\beta_N&=-\frac{n^2}{2}(\omega+\frac{1}{ \omega}+\frac{2g}{ \omega^2}+O(\frac{1}{n^3})).
\label{betatilden}
\end{align}
The condition $\lim_{n \rightarrow \infty} A^{(N-)}_{ab}=0$ (or equivalently $\lim_{n \rightarrow \infty} \frac{\partial \beta_N}{\partial \omega}=0$) implies the equation
\begin{equation}
\omega^3=\omega+4 g,
\label{opt freq}
\end{equation}
whose real solution $\tilde{\omega}$ is obviously \emph{different} than the natural choice $\omega=1$ suggested by the original quadratic term. As we will see in Section \ref{perturbations diagrams}, $\tilde{\omega}$ provides a crude approximation of the spectral gap for this model. According to Lemma \ref{lemma1}, the Hamiltonian in the optimized Fock space $\mathcal{H}_{\tilde{\omega}}$ takes the following form
\begin{align}
 2H_N=  2\sum_{a} \tilde{\omega} a_a^{\dagger}a_a + \frac{2 g}{n} :Tr{M^4}:+ n^2 e_{0}^{(0)} \mathbb{I}+R_N,
 \label{quantum anharmonic}
\end{align}
where $||R_N \psi|| \rightarrow 0~~ \forall \psi \in \mathcal{H}_{\tilde{\omega}}$. 
Inserting $\tilde{\omega}$ to (\ref{betatilden}) gives the gaussian variational upper bound for the ground state energy
\begin{equation}
\frac{e_0^{(0)}(g)}{2}=\frac{\tilde{\omega}}{4}+\frac{1}{4 \tilde{\omega}}+\frac{g}{2 \tilde{\omega}^2},
\label{gaussian bound AO}
\end{equation}
which is in an excellent agreement with the result of \cite{planar} (even for a large coupling $g$), where the authors obtained the exact value. Asymptotically, for large $g$,  they have $e_0(g) \simeq 0.58993 g^{\frac{1}{3}}$. Our variational bound $\frac{e_0^{(0)}(g)}{2}\simeq 0.59527 g^{\frac{1}{3}}$ is at most $\approx 9 \permil$ wrong (see Table \ref{table1}).

\subsection{Spectral bounds}
\label{spectral bounds AO}
In order to produce a lower bound for the spectrum, one has to take into account the interaction term $\frac{1}{n}:V_N:$. This involves quite technical, but instructive estimates of matrix elements of the Hamiltonian between $U(n)$ invariant wave functions, which we present in this subsection.

Proceeding along the lines of \cite{Fock Space methods} we introduce a basis of the $U(n)$ invariant subspace $\mathcal{I}_{\tilde{\omega}}^{(n)} \subset \mathcal{H}_{\tilde{\omega}}$ spanned by $U(n)$-invariant linear combinations of the first $N=n^2$ creation operators, called the partitions basis 
\begin{equation}
\psi_{\lambda}:= \mathcal{N}_{\lambda} (a^{\dagger})^{\lambda} \Psi_0(\tilde{\omega}) \equiv \mathcal{N}_{\lambda} Tr(a^{\dagger \lambda_1}) Tr(a^{\dagger \lambda_2})...Tr(a^{\dagger \lambda_m}) \Psi_0(\tilde{\omega}), \label{partitions basis}
\end{equation}
where $Tr(a^{\dagger \lambda_i}):=Tr(T_{b_1}...T_{b_{\lambda_i}})a^{\dagger}_{b_1}...a^{\dagger}_{b_{\lambda_i}}$ and $\lambda=(1^{\lambda_1},2^{\lambda_2},...,m^{\lambda_m})$ is a partition of a certain natural number $k=|\lambda|:= \sum_i i \lambda_i$. The partitions basis becomes orthonormal at $N= \infty$ for the properly chosen normalisation factors $\mathcal{N}_{\lambda} \propto n^{-\frac{|\lambda|}{2}}$ (generically), see \cite{33}. 

As shown in \cite{Fock Space methods}, the matrix elements of the Hamiltonian (\ref{quantum anharmonic}) in the partitions basis contain three groups of divergent terms 
\begin{enumerate}
\item the vacuum expectation value $(\Psi_0, H_N \Psi_0)\propto n^2$ corresponding to our $\beta_N \mathbb{I}$,
\item $(\psi_{\lambda}, H_N \psi_{\delta})\propto n$, where $\lambda_2=\delta_2 \pm 1$, coming from the $Tr(a^{\dagger 2}+a^2)$ part of $H_N$,
\item $(\psi_{\lambda}, H_N \psi_{\delta})\propto n$, where $\lambda_4=\delta_4 \pm 1$, coming from the $\frac{1}{n} Tr(a^{\dagger 4}+a^4)$ part of $H_N$.
\end{enumerate}
Renormalization of $H_N$ is based on a proper choice of basis (in particular the vacuum) such that the second and third group of divergent matrix elements would be "absorbed" into the first one as a constant shift of the whole spectrum. The divergent vacuum energy can be then easily subtracted.
According to Lemma \ref{lemma1}, the proper choice of the Fock space frequencies allows to absorb Group 2 into the ground state energy by eliminating the $Tr(a^{\dagger 2}+a^2)$ part from the description. Then, for the suitable $\beta_N$, the only divergent matrix elements of $H_N+\beta_N \mathbb{I}$ belong to Group 3 but they are much more difficult to handle since the resulting vacuum is no longer a standard Fock vacuum (there are no counterterms which would cancel $Tr(a^{\dagger 4}+a^4)$ after an appropriate choice of the Fock space frequencies as it happened with $Tr(a^{\dagger 2}+a^2)$ ). In order to approach this problem let us point out that in the planar limit one can interpret $Tr(a^{\dagger 4})$ and $ Tr(a^4)$ as composite creation-annihilation operators 
\begin{align}
A&:=Tr(T_a T_b T_c T_d) a_a a_b a_c a_d \\
\left[A, A^{\dagger} \right]&= 4 n^4 \mathbb{I}+O(n^2).
\end{align}
Intuition suggests that one should treat $A$ and $A^{\dagger}$ in a special way in the Fock space description. This can be realized as follows. Let us introduce a new basis, completely equivalent to the partitions basis, consisting of
\begin{equation}
\psi_{\lambda, k}:= \mathcal{N}_{k,\lambda} (a^{\dagger})^{\lambda} (A^{\dagger})^k \psi_0,~ k=0,1,2...,~ \lambda_4=0,
\label{Abasis}
\end{equation}
where $ \mathcal{N}_{k,\lambda} \propto  \mathcal{N}_{\lambda}n^{-2k}$. It turns out that the states $\psi_{\lambda,k}$ exhibit very useful properties, which we explore below.

\begin{lemma} 
\label{lemma2}
The action of the operator
\begin{equation}
H_N= \alpha a_a^{\dagger} a_a+ \frac{\beta}{n}(A+A^{\dagger} + \gamma Tr(a^{\dagger}a^{\dagger} a a)) 
\label{operator1}
\end{equation}
on the states $\psi_{ \lambda k}$, asymptotically \footnote{ All operator inequalities are meant here  in the usual sense, i.e. $H_1 \geq H_2$ iff $\langle \psi,H_1 \psi \rangle \geq \langle \psi,H_2 \psi \rangle~~ \forall \psi \in Dom(H_1) \subset Dom(H_2)$. Also, terms of order $O(\frac{1}{n})$ are meant in the sense of norm in $\mathcal{H}_{\tilde{\omega}}$ and asymptotically equal terms differ by terms of order at most $O(\frac{1}{n})$. We call an operator $T_N$ to be of order $O(n^k)$ iff $||T_N \psi|| \leq const.(\psi) n^k, \forall \psi \in \mathcal{H}_{\omega} $. The key observation which allows to compute the leading terms at large $n$ (i.e. of order $O(1)$, $O(n)$ and $O(n^2)$ ) is the fact that they originate from Wick contractions of adjacent operators sitting in one $U(n)/SU(n)$ trace, resp. corresponding to planar contractions in the diagramatic representation (see Section \ref{perturbations diagrams}) and we refer to the large $n$ limit taken by neglecting all non-planar contributions, resp. all subleading terms (i.e. of order $O(\frac{1}{n})$ and lower) as the planar limit. In Section 5 we give a perturbative justification of this limit up to the third order}, at large $n$, becomes
\begin{equation}
H_N \psi_{\lambda,k}= \left( \frac{\Omega}{n^4} B^{\dagger} B+ (\tilde{e}_0 n^2 +G(\lambda))\mathbb{I}  \right) \psi_{\lambda,k}
\end{equation}
with $G(\lambda)= \alpha|\lambda| + \beta \gamma \sum_{i=2} i \lambda_i $, $\Omega= \alpha + \beta \gamma$ and $\tilde{e}_0= -\frac{\beta^2}{\alpha + \beta \gamma}$ and $ B=A+\frac{\beta n^3}{\alpha+ \beta \gamma} \mathbb{I}$
\end{lemma}

\begin{proof}
First we prove that 
\begin{equation}
a^{\dagger}_a a_a \psi_{\lambda k} \simeq \left(  \frac{1}{n^4} A^{\dagger} A + |\lambda| \right)\psi_{\lambda,k}
\label{adaggera_AdaggerA exchange}
\end{equation}
Indeed, $a^{\dagger}_a a_a \psi_{\lambda k}= (4 k+|\lambda|)\psi_{\lambda k}$ and 
\begin{align}
 \frac{1}{n^4} A^{\dagger} A \psi_{\lambda k} =  \frac{\mathcal{N}_{k,\lambda}}{n^4} A^{\dagger} [A, (a^{\dagger})^{\lambda} (A^{\dagger})^k] A \psi_{0} 
 = \frac{\mathcal{N}_{k,\lambda}}{n^4}  \left(A^{\dagger} [A,(a^{\dagger})^{\lambda}](A^{\dagger})^k  + A^{\dagger}(a^{\dagger})^{\lambda}[A,(A^{\dagger})^k]  \right) \psi_0\\
 = \frac{\mathcal{N}_{k,\lambda}}{n^4}  A^{\dagger} [A,(a^{\dagger})^{\lambda}](A^{\dagger})^k  \psi_0 + 4k \psi_{\lambda,k}+O(\frac{1}{n^2})\label{first term}
\end{align}
One has to show that the first term in (\ref{first term}) converges to $0$ in norm at large  $n$ $\forall \lambda$. According to Wick's theorem we have 4 cases 
\begin{itemize}
\item \emph{single contractions}:   we get $n^{-\frac{1}{2}}$ from the fact that the partition $\lambda$ has been shortened by 1. Then we are left with three annihilation operators acting on $(A^{\dagger})^k  \psi_0$, which prolongs $\lambda $ by 1 and thus gives a factor of $n^{\frac{1}{2}}$ as well as $n^2$ (at most, when we perform a planar contraction, resp. when adjacent indices of $U(N)$ are contracted) from the contraction with $A^{\dagger}$. The total factor is then $n^{-2} \rightarrow 0$;
\item \emph{double contractions}:  we get $n^1$ (at most) from the double contraction, $n^{-1}$ from the fact that the partition $\lambda$ has been shortened by 2. Then we are left with two annihilation operators acting on $(A^{\dagger})^k  \psi_0$, which prolongs $\lambda $ by 2 and thus gives a factor of $n^1$ as well as one more  $n^1$ (at most) from the contraction with $A^{\dagger}$. The total factor is then $n^{-2} \rightarrow 0$;
\item \emph{triple contractions}:  we get $n^2$ (at most) from the triple contraction, $n^{-\frac{3}{2}}$ from the fact that the partition $\lambda$ has been shortened by 3. Then we are left with one annihilation operator acting on $(A^{\dagger})^k  \psi_0$, which prolongs $\lambda $ by 3 (and thus gives a factor of $n^{\frac{3}{2}}$). The total factor is $n^{-2} \rightarrow 0$;
\item \emph{quadrupole contractions}:  we get $n^3$ (at most) from the quadrupole contraction, $n^{-2}$ from the fact that the partition $\lambda$ has been shortened by 4 and $n^2$  coming from the fact that $k$ has increased by 1. The overall factor becomes then $n^{-4} n^3 n^{-2} n^2= n^{-1} \rightarrow 0 $,
\end{itemize}
which proves (\ref{adaggera_AdaggerA exchange}). Then we observe (using a similar justification as above) that at large $n$ we get
\begin{align}
\frac{1}{n}(a^{\dagger}a^{\dagger} a a)\psi_{\lambda,k} \simeq \left(\sum_{k=2}^{\infty} \lambda_k + 4 k \right)\psi_{\lambda,k}\simeq \left(\sum_{k=2}^{\infty} \lambda_k + \frac{1}{n^4}  A^{\dagger} A ,\right)\psi_{\lambda,k}.
\end{align}
After introducing a new operator $B=A+\frac{\beta n^3}{\alpha+ \beta \gamma} \mathbb{I}$, we obtain
\begin{equation}
H_N \psi_{\lambda,k}= \left( \frac{\Omega}{n^4} A^{\dagger} A+ \frac{\beta}{n}(A+A^{\dagger}) +G(\lambda)\mathbb{I}  \right) \psi_{\lambda,k}=\left( \frac{\Omega}{n^4} B^{\dagger} B+ (\tilde{e}_0 n^2 +G(\lambda))\mathbb{I}  \right) \psi_{\lambda,k}.
\end{equation}
\end{proof}
Now we come back to the full Hamiltonian, which can be rewritten as 
\begin{align}
H_N=\tilde{\omega}  (1- \epsilon) a_a^{\dagger} a_a +\frac{g}{4 n \tilde{\omega}^2} (A+A^{\dagger}+4 Tr(a^{\dagger} a^{\dagger} a a))\label{Hneps1}\\
+ \epsilon \tilde{ \omega}  a_a^{\dagger} a_a+ \frac{g}{n\tilde{\omega}^2} (Tr(a^{\dagger} a^{\dagger} a^{\dagger} a)+Tr(a^{\dagger} a a a))\label{Hneps2}\\
+\frac{g}{ 2 n \tilde{\omega}^2}:Tr(a^{\dagger}a a^{\dagger}a):+\frac{e_0^{(0)}}{2}n^2 \mathbb{I},\label{Hneps3}
\end{align}
for some $0< \epsilon <1$. The first part of the Hamiltonian i.e. (\ref{Hneps1}) is exactly of the form considered in Lemma \ref{lemma2}, while the first term in (\ref{Hneps3}) is of order $O(\frac{1}{n})$, hence it does not contribute to the planar limit, see Section \ref{perturbations diagrams}. Now we will determine for which values of $\epsilon$ the middle part of $H_N$ namely (\ref{Hneps2}) is asymptotically non-negative, which will allow to bound the whole Hamiltonian from below by an operator of the form (\ref{Hneps1}).

\begin{lemma}
Let $\omega, \delta >0$. Then the operator
\begin{equation}
H_{n}= \omega a_a^{\dagger} a_a+ \frac{\delta}{n}\left( Tr( a^{\dagger} a^{\dagger} a^{\dagger} a)+Tr(a^{\dagger} aaa) \right)
\end{equation}
is non-negative definite at large $n$ if $\frac{\omega}{\delta} > 2$.\\

\end{lemma}
\begin{proof}
Consider the following non-negative definite operator ($(abc):=Tr(T_a T_b T_c)$)
\begin{equation}
(\alpha (\epsilon a b)a_a a_b+ \beta (\epsilon a b) a_a^{\dagger}a_b)^{\dagger} (\alpha (\epsilon c d)a_c a_d + \beta (\epsilon c d ) a_c^{\dagger} a_d) \geq 0,~~~~\alpha, \beta \in \mathbb{R}.
\end{equation}
After multiplying the parentheses, using the completeness relation (\ref{completeness relation1}), and noting that $(a^{\dagger}a) \geq \frac{1}{n}(a^{\dagger}a^{\dagger} aa) + O(\frac{1}{n})$ we get 
\begin{align}
-\frac{\alpha \beta}{n}((a^{\dagger}a a a )+(a^{\dagger}a^{\dagger}a^{\dagger}a)) \leq \frac{1}{n}(\alpha^2 + \beta^2 )(a^{\dagger}a) +O(\frac{1}{n})
\end{align}
By taking $-\alpha \beta<0$ and choosing the optimal constants $\alpha, \beta$ we get
\begin{equation}
\frac{1}{n}(Tr( a^{\dagger} a^{\dagger} a^{\dagger} a)+Tr(a^{\dagger} aaa)) \geq -2 Tr(a^{\dagger} a)+O(\frac{1}{n})
\end{equation} 
and thus $H_n \geq \delta (\frac{\omega}{\delta}-2) a_a^{\dagger} a_a +O(\frac{1}{n}) \geq O(\frac{1}{n})$ provided $\frac{\omega}{\delta} > 2$.
\end{proof}

Using Lemma 2 with $\alpha=\tilde{ \omega} (1 - \epsilon)$, $\beta =\frac{g}{4 \tilde{\omega}^2}$ and $\gamma=4$ and Lemma 3 with $\omega= \tilde{ \omega}\epsilon $, $\delta= \frac{g}{ \tilde{\omega}^2} $ and $\epsilon >\frac{g}{\tilde{\omega}^3}$ we get quite a useful bound.

 \begin{theorem}
 The Hamiltonian of the $U(n)$ invariant anharmonic oscillator (\ref{1matrix}) is bounded below at large n by the following operator
 \begin{align}
H_{N}^{-}:= \frac{\Omega}{n^4} B^{\dagger} B+  \sum_{\lambda} G(\lambda) P_{\lambda}+ \frac{1}{2} \left(e_0^{(0)}+ \tilde{e}_0\right) n^2 \mathbb{I}  +O(\frac{1}{n})
\end{align}
with $G(\lambda)= \alpha|\lambda|  + 4 \beta  \sum_{i=2} i \lambda_i $, $\Omega= \alpha + 4 \beta $ and $\tilde{e}_0= -\frac{2 \beta^2}{\alpha + 4 \beta }$, $\alpha=\tilde{ \omega}(1-\epsilon)$, $\beta =\frac{g}{4 \tilde{\omega}^2}$,  $ B=A+\frac{\beta n^3}{\alpha+ 4 \beta} \mathbb{I}$ and $\epsilon >\frac{g}{\tilde{\omega}^3}$, where $P_{\lambda}$ is the orthogonal projection onto the subspace $span(\psi_{k,\lambda}),~k=0,1,...$ with $\lambda_4=0$.
 \end{theorem}

 Since $H_{N}^{-} \geq \frac{1}{2} \left(e_0^{(0)}+ \tilde{e}_0\right)  n^2 \mathbb{I}  +O(\frac{1}{n})$, we can easily find a lower bound for $\frac{e_0}{2}$ (and for the whole spectrum) in the planar limit

\begin{align}
\frac{e_0^{(lower)}}{2}:=\frac{e_0^{(0)}}{2} -\frac{\beta^2}{\alpha+4 \beta}.
\end{align}
This also shows the mechanism how the $\frac{1}{n} Tr(a^{\dagger 4} + a^4)$ part of the Hamiltonian (so terms of order $O(n)$) shifts the whole spectrum by  $const. n^2$.
 We refer the reader to Table \ref{table1} for several numerical values.\\
 \begin{table}[H]
\caption{Comparison of the exact ground state energies $e_0$ \cite{planar} with the optimized Fock space ground state energy (the gaussian upper bound $e_0^{(0)}/2$) and the lower bound $ e_0^{(lower)}/2$ from Section \ref{spectral bounds AO} for several values of the coupling constant $g$, at $n= \infty$.}
\label{table1}
\begin{center}
    \begin{tabular}{ | l | l | l |l |l |l | p{5cm} |}
    \hline
g & $e_0^{(0)}/2$& $ e_0^{(lower)}/2 $ & $e_0$  \\ \hline
 0.01 &0.505 & 0.505 & 0.505  \\ \hline
 0.1 & 0.543&0.542&0.542 \\        \hline
0.5 &0.653&   0.651&  0.651  \\ \hline
1.0 &  0.743&   0.740 &  0.740  \\  \hline
50 & 2.235& 2.214& 2.217  \\        \hline
1000 & 5.968 & 5.907& 5.915   \\        \hline
$g \rightarrow \infty $  & 0.59527 $g^{\frac{1}{3}}$ & 0.589075 $g^{\frac{1}{3}}$   & 0.58993 $g^{\frac{1}{3}}$   \\        \hline
    \end{tabular}
    \end{center}
    \end{table}

\section{Membrane Matrix Models}
\label{chapter mmm}
We will start with the optimized Fock space decomposition for the MMM (\ref{generalresc}) and obtain an upper bound for the ground state energy for arbitrary $d$. Later, we will produce a lower bound for the spectrum based on a direct generalization of Theorem 1 (and the preceding lemmas) to the multi-matrix case. We restrict the Hamiltonian to the constrained by (\ref{constraints}) subspace of $\mathcal{H}_{\omega}$ i.e. to the $SU(n)\times SO(d)$ invariant sector and thus it is natural to assume that the optimized Fock space frequencies are all the same $ \tilde{\omega}_{I} \equiv \tilde{\omega}_{a_I i_I}= \tilde{\omega} ~~ \forall I,  I=(a_I,i_I),~~a_I=1,...,n^2-1,~~i_I=1,...,d$ and consider only Fock spaces with all the frequencies equal. 
 Using  
 \begin{equation}
  c_{IJKL} =\frac{1}{2} f_{a a_I a_K}f_{a a_J a_L}\delta_{i_I i_J}\delta_{i_K i_L}
 \end{equation}
 
 with $f^{(n)}_{abc}=\frac{2 \pi n^{\frac{3}{2}}}{i}Tr(T_a \left[T_b,T_c \right])$ and the completeness relation 
\begin{equation}
 (T_a)_{ij}(T_a)_{kl}=\delta_{jk} \delta_{il}-\frac{1}{n}\delta_{ij} \delta_{kl},
 \label{completeness relation}
\end{equation}  
  we get 
 \begin{align}
 \sum_{K}\frac{c^{(N)}_{KK IJ}}{\omega_K \sqrt{\omega_I \omega_J}}=\sum_{K}\frac{c^{(N)}_{ IJ KK}}{\omega_K \sqrt{\omega_I \omega_J}}=(2 \pi)^2  n^4 \delta_{a_I a_J}\delta_{i_I i_J}\frac{d}{\omega^2},\\
\sum_{K}\frac{c^{(N)}_{K IJK}}{\omega_K \sqrt{\omega_I \omega_J}}=\sum_{K}\frac{c^{(N)}_{IKKJ}}{\omega_K \sqrt{\omega_I \omega_J}}=-(2 \pi)^2  n^4 \delta_{a_I a_J}\delta_{i_I i_J}\frac{1}{\omega^2},\\
c^{(N)}_{KIKJ}=c^{(N)}_{IKJK}=0,\\
\sum_{IJ}\frac{c^{(N)}_{II JJ}}{\omega_I \omega_J}=(2 \pi)^2 n^4(n^2-1)\frac{d^2}{\omega^2},\\
\sum_{IJ} \frac{c^{(N)}_{I JJI}}{\omega_I \omega_J}=-(2 \pi)^2 n^4(n^2-1)  \frac{d}{\omega^2},\\
\end{align}
and therefore
\begin{align}
f(n)=\frac{\pi^2 d(d-1)(n^2-1)}{\omega^2} n^4,\\
A_{IJ}^{(N)}=\frac{4\pi^2 (d-1)}{\omega^2} n^4 \delta_{IJ},~~~\gamma_{crit}=2,\\
A_{IJ}^{(N \pm)}=\frac{4\pi^2 (d-1)}{\omega^2} \frac{n^4}{(n^2-1)^2} \delta_{IJ} \pm \omega \delta_{IJ}.
\label{MMM tensors}
\end{align}
Then the condition (\ref{diagonallimit}) is equivalent to
 \begin{align}
\omega =\lim_{n \rightarrow \infty } (n^2-1)^{-2}A^{(N)}_{II}=\lim_{n \rightarrow \infty } n^{-4}\sum_{K}\frac{c^{(N)}_{(I I KK)}}{8 \omega_K \omega_I} 
= (2 \pi)^2  \frac{1}{\omega^2}(d-1),
\label{optimized frequencies MMM}
\end{align} 
which shows that $\tilde{\omega} = \sqrt[3]{4 \pi^2 (d-1)}$. Therefore for the choice $\omega=\tilde{\omega}, ~\beta_N= -\frac{1}{2}d(n^2-1)\tilde{\omega}-\frac{\pi^2 d(d-1)(n^2-1)}{\tilde{\omega}^2}$ the renormalized Hamiltonian becomes
\begin{align}
H_N+\beta_N \mathbb{I}=2 \tilde{\omega} \sum_{i,a} a_{ia}^{\dagger}a_{ia}  +n^{-4}:V_N:+R_N\\
=2 \tilde{\omega} \sum_{i,a} a_{ia}^{\dagger}a_{ia} +\frac{4 \pi^2}{n}\left(Tr(abcd)-Tr(acbd)\right):x_{ia}x_{ib}x_{jc}x_{jd}:+R_N,
\label{quantum membrane}
\end{align}
where $R_N \propto  \frac{1}{n^2}\sum_{i,a} (a_{ia}^{\dagger}a_{ia}^{\dagger}+a_{ia}a_{ia})\propto O(\frac{1}{n^2})$, hence  $||R_N \psi|| \rightarrow 0~~ \forall \psi \in \mathcal{H}_{\tilde{\omega}}$, and the ground state energy for the optimized Fock space approximation (i.e. the gaussian upper bound) is given by 
\begin{equation}
 e_0^{(0)}= -\lim_{n\rightarrow \infty} \frac{\beta_N}{n^2}=\frac{ \tilde{\omega}d }{2}+\frac{\pi^2d (d-1)}{\tilde{\omega}^2}=\frac{3 d(d-1)^{\frac{1}{3}}\pi ^{\frac{2}{3}}}{ 2^{\frac{4}{3}}}.
 \label{gaussian bound MMM}
\end{equation}
Let us point out that the scaling given by $\gamma_{crit}$ leads in fact to 't Hooft's scaling, since the quartic potential in the rescaled Hamiltonian (\ref{quantum membrane}) is multiplied by $\frac{1}{n}$. Moreover, this is the only scaling for which the optimized Fock space frequency is non-trivial, i.e. $0< \tilde{\omega} < \infty$ (cp. (\ref{opt freq}) for the AO, where 't Hooft's coupling is the only one leading to a non-trivial redefinition of the optimized mass/frequency). As we will see in Section \ref{perturbations diagrams}, $4\tilde{\omega}$ provides a crude approximation for the mass gap in the $SU(\infty) \times SO(d) $ invariant sector.

\subsection{Spectral bounds}
\label{spectral bounds MMM}
In this section we will present estimates for matrix elements of various parts of the Hamiltonian (\ref{quantum membrane}) leading to a lower bound for the ground state energy in the planar limit, which is a direct generalisation of the procedure introduced in Section \ref{spectral bounds AO} for the AO. 
In order to obtain a lower bound for the spectrum, one has to take into account the quartic interaction term $n^{-4}:V_N:$ in the $SO(d) \times SU(n)$ invariant sector. As for the AO, one can show that the only divergent matrix elements of $H_N+\beta_N \mathbb{I}$ (\ref{quantum membrane}) in the appropriately modified partitions basis of the space of invariants are those coming from the operators $Tr(a^{\dagger 4}), Tr(a^4)$ due to the fact that the optimal choice of the Fock space frequencies (\ref{optimized frequencies MMM}) eliminates the $Tr(a^{\dagger 2}+a^2)$ part from the game.
As previously, a special treatment of the $Tr(a^{\dagger 4}+a^4)$ part of $H_N$ is the main point of our construction.

Let us introduce two composite annihilation operators
\begin{align}
A&:= Tr(abcd)a_{ai}a_{bi}a_{cj}a_{dj} \equiv (iijj),\\
B&:=Tr(abcd)\left(\frac{1}{d+1}  a_{ai}a_{bi}a_{cj}a_{dj}-\frac{1}{2}  a_{ai}a_{bj}a_{ci}a_{dj}\right)
\equiv \frac{1}{d+1} (iijj)-\frac{1}{2}  (ijij),\\
\left[A, A^{\dagger} \right]&=2 d n^4 (d+1)+O(n^2),\\
\left[B, B^{\dagger} \right]&=\frac{d(d+2)(d-1)}{d+1}n^4+O(n^2),\\
\left[A, B^{\dagger} \right]&=\left[B, A^{\dagger} \right]=O(n^2),
\end{align}
and the following $SU(N)\times SO(d)$ invariant states, keeping the dependence on $A$ and $B$ explicit, cp. (\ref{Abasis}),
\begin{equation}
\psi_{kl\Lambda}:= \mathcal{N}_{kl\Lambda}(a^{\dagger})^{\Lambda}(A^{\dagger})^k (B^{\dagger})^l \psi_0,~~ \mathcal{N}_{kl\Lambda} \propto n^{-\frac{|\lambda_{\Lambda}|+4k+4l}{2}}
\label{partitions basis2}
\end{equation}
where $\Lambda= \lbrace \lambda_{\Lambda}, I_{\Lambda}\rbrace$ is a partition of some $k \in \mathbb{N}$ equipped with the $SO(d)$ invariant structure of  the state $\psi_{kl\Lambda}$, i.e. $\lambda_{\Lambda} \vdash k$. $I_{\Lambda}$ is a sequence containing the information about $SO(d)$ indices "compatible"  with the $SU(n)$ partition structure\footnote{The Hamiltonian asymptotically stabilizes a subspace of $SU(n) \times SO(d)$ invariants, where contractions between $SO(d)$ indices do not occur between different $SU(n)$ traces. One way to see this is to use (\ref{consise potential1})-(\ref{consise potential}) and the completeness relation (\ref{completeness relation}) and note that the only terms violating the invariance of this subspace are those coming either from the "$\frac{1}{n}$ part" of the completeness relation or from contractions between non-adjacent indices producing subleading terms in $n$.} , e.g. for $\lambda=  (2^2)$, i.e. $\lambda_2=2, \lambda_i=0, ~i\neq 2 $, a possible $I$ could be $I=(ii,jj)$, but not $I=(ij,ij)$. The corresponding states $\psi_{kl\Lambda}$ would be then
\begin{align}
\psi_{kl\Lambda}= \mathcal{N}_{kl\Lambda}Tr(ab)Tr(cd) a^{\dagger}_{ia} a^{\dagger}_{ib} a^{\dagger}_{jc} a^{\dagger}_{jd}(A^{\dagger})^k (B^{\dagger})^l \psi_0, ~~~I_{\Lambda}=(ii,jj),
\end{align}
and the not allowed state
\begin{align}
\psi_{kl\Lambda}= \mathcal{N}_{kl\Lambda}Tr(ab)Tr(cd) a^{\dagger}_{ia} a^{\dagger}_{jb} a^{\dagger}_{ic} a^{\dagger}_{jd}(A^{\dagger})^k (B^{\dagger})^l \psi_0, ~~~I_{\Lambda}=(ij,ij).
\end{align}
We denote the subspace spanned by $\psi_{kl \Lambda}$ by $\mathcal{I}_{\tilde{\omega}}^{(n)}$. Note that the basis $(\ref{partitions basis2})$ is not orthonormal and at large $n$ we have the orthogonality relation
\begin{equation}
\langle \psi_{kl\Lambda}, \psi_{k'l'\Lambda'}\rangle=\delta_{k k'}\delta_{l l'} \delta_{\lambda_{\Lambda} \lambda_{\Lambda'}} G(I_{\Lambda},I_{\Lambda'}),
\end{equation}
where $G(I_{\Lambda},I_{\Lambda'})$ is the Gram matrix of (\ref{partitions basis2}) restricted to $span(\psi_{kl\Lambda})$ with $k,l,\lambda_{\Lambda}$ fixed. 

The normal ordered quartic potential in (\ref{quantum membrane}) can be rewritten as

\begin{align}
V_n= \frac{\pi^2}{n \tilde{\omega}^2}\left( (i^{\dagger}i^{\dagger}j^{\dagger}j^{\dagger})+(iijj) -(i^{\dagger}j^{\dagger}i^{\dagger}j^{\dagger}) - (ijij)\right)\label{consise potential1}\\
+ \frac{2\pi^2}{n \tilde{\omega}^2}\left( (i^{\dagger} i^{\dagger} jj) + (i^{\dagger} j^{\dagger} ji)-2(i^{\dagger} j^{\dagger} ij)
-:(i^{\dagger} j i^{\dagger} j):+\frac{1}{2}:(i^{\dagger} i j^{\dagger} j):+\frac{1}{2}:(i^{\dagger} j j^{\dagger} i)):\right) \\+
\frac{2 \pi^2}{n \tilde{\omega}^2}\left((i^{\dagger}i^{\dagger}j^{\dagger}j)+(j^{\dagger}i^{\dagger}i^{\dagger}j)-2(i^{\dagger}j^{\dagger}i^{\dagger}j)+
(i^{\dagger}ijj)+(j^{\dagger}iij)-2(i^{\dagger}jij)
\right),
\label{consise potential}
\end{align}
and therefore it is useful to introduce the following operators asymptotically stabilizing the subspaces $ span(\psi_{kl\Lambda},\psi_{k+1, l-1,\Lambda},\psi_{k-1, l+1,\Lambda})_{\lambda_{\Lambda}=\lambda} $ (which means that they only mix up the $SO(d)$ indices and not the $SU(n)$ partition structure)
\begin{align}
\frac{1}{n}  S_1:= \frac{1}{n} Tr(abcd)a_{ai}^{\dagger}a^{\dagger}_{bi}a_{cj}a_{dj} \equiv \frac{1}{n}(i^{\dagger} i^{\dagger} j j),\\
\frac{1}{n} S_2:=\frac{1}{n} Tr(abcd)a_{ai}^{\dagger}a^{\dagger}_{bj}a_{cj}a_{di}\equiv  \frac{1}{n} (i^{\dagger} j^{\dagger} j i),\\
\frac{1}{n} S_3:=\frac{1}{n} Tr(abcd)a_{ai}^{\dagger}a^{\dagger}_{bj}a_{ci}a_{dj} \equiv \frac{1}{n} (i^{\dagger} j^{\dagger} i j),\\
\end{align}
as well as
\begin{align}
\frac{1}{n} T_1:=\frac{1}{n}Tr(abcd)a_{ai}^{\dagger}a_{bi}a_{cj}a_{dj} \equiv \frac{1}{n}(i^{\dagger} i j j),\\
\frac{1}{n} T_2:=\frac{1}{n}Tr(abcd)a_{ai}^{\dagger}a_{bj}a_{cj}a_{di}\equiv \frac{1}{n}(i^{\dagger} jj i),\\
\frac{1}{n} T_3:=\frac{1}{n}Tr(abcd)a_{ai}^{\dagger}a_{bj}a_{ci}a_{dj} \equiv \frac{1}{n}(i^{\dagger} j i j),\\
\end{align}
which change also the $SU(n)$-partition structure (the $\lambda_{\Lambda}$ part). Now we come back to the full Hamiltonian, which can be rewritten in terms of the recently introduced operators
\begin{align}
H_N=2 (1-\epsilon)  \tilde{\omega}(i^{\dagger} i) +\frac{2\pi^2}{\tilde{\omega}^2n}\left(\frac{d-1}{2(d+1)}(
A+A^{\dagger})+B+B^{\dagger}+ S_1+S_2-2S_3\right)\label{ABSoperator} \\
+  2 \epsilon \tilde{\omega}\left[ (i^{\dagger} i)+ \frac{ \pi^2}{ n \tilde{\omega}^3 \epsilon }\left(T_1+T_1^{\dagger}+T_2+T_2^{\dagger}-2(T_3+T_3^{\dagger})\right)\right] \label{Toperator}\\
+ \frac{\pi^2}{n \tilde{\omega}^2}:\left( 
(i^{\dagger} i j^{\dagger} j)+(i^{\dagger} j j^{\dagger} i))-2(i^{\dagger} j i^{\dagger} j)\right):+e^{(0)}_0 n^2 \mathbb{I},\label{lastpart}
\end{align}
for some $0<\epsilon <1$, cp. (\ref{Hneps1})-(\ref{Hneps3}). Following the strategy established for the AO, we will express the action of the first part of the Hamiltonian, i.e. (\ref{ABSoperator}) in terms of $A,B$ and the representation of $S_1+S_2-2S_3$ in the finite dimensional blocks $P_{\lambda,k,l} \mathcal{I}_{\tilde{\omega}}:=span(\psi_{kl \Lambda})_{\lambda_{\Lambda}=\lambda}$, then we will determine for which value of $\epsilon$ the middle part (\ref{Toperator}) is non-negative, which will allow to bound the full Hamiltonian from below by an operator of the form as in (\ref{ABSoperator}) modulo terms of order $O(\frac{1}{n})$ which do not affect the divergent vacuum energy in the planar limit. Let us perform this in several steps.

\begin{lemma}
\label{lemma5}
The action of the operator
\begin{equation}
H_N= \omega a_{ia}^{\dagger} a_{ia}+ \frac{1}{n}(\epsilon_1(A+A^{\dagger}) + \epsilon_2(B+B^{\dagger}) +g(S_1+S_2-2S_3) ) 
\label{operator11}
\end{equation}
on the states $\psi_{k l \Lambda }$, asymptotically at large $n$, becomes
\begin{equation}
H_N \psi_{ k l \Lambda}= \left[\Omega_1  \tilde{A}^{\dagger}\tilde{A}+\Omega_2  \tilde{B}^{\dagger} \tilde{B} +\Omega_3( \tilde{B}^{\dagger}\tilde{A}+ \tilde{A}^{\dagger}\tilde{B})+(\tilde{e}_0 n^2 +g G(\lambda_{\Lambda})+\omega |\lambda_{\Lambda}|)\mathbb{I}  \right] \psi_{ k l \Lambda},
\end{equation}
where $G(\lambda_{\Lambda})$ is the finite dimensional representation of $S_1+S_2-2S_3$ in  $P_{\lambda_{\Lambda},k,l} \mathcal{I}_{\tilde{\omega}}$, 
\begin{align}
\Omega_1=\Omega_1(\omega,g)&=\frac{g}{n^4}\left( \frac{1}{d} +\frac{2}{d(d+1)}-\frac{2(d+3)}{d(d+1)^2} \right)+\frac{2 \omega}{d(d+1) n^4},\\
\Omega_2=\Omega_2(\omega,g)&=\frac{4g(d+3)}{d (d-1)(d+2) n^4}+\frac{4 \omega(d+1)}{d(d-1)(d+2)n^4}, \\
\Omega_3=\Omega_3(\omega,g)&=\frac{4g}{d(d+1) n^4 },\\
n^2 \tilde{e}_0=n^2\tilde{e}_0(\omega,\epsilon_1,\epsilon_2,g)&= \Omega_1 \alpha^2 + \Omega_2 \beta^2 + 2 \Omega_3 \alpha \beta + 
 2 \frac{\epsilon_1 \alpha}{n} + 2 \frac{\epsilon_2 \beta}{n},
\end{align}
and $\tilde{A}=A+\alpha \mathbb{I},~~\tilde{B}=B+\beta \mathbb{I}$ with
\begin{equation}
\label{matrixcondition}
\left(
\begin{array}{c}
 \alpha  \\
 \beta  \\
\end{array}
\right)=-\frac{1}{n}\left(
\begin{array}{cc}
 \Omega_1  & \Omega_3  \\
 \Omega_3  & \Omega_2  \\
\end{array}
\right)^{-1} \left(
\begin{array}{c}
 \epsilon_1  \\
 \epsilon_2  \\
\end{array}
\right).
\end{equation}
\end{lemma}
\begin{proof}

Repeating the argument from Lemma 1, as the $SO(d)$ structure does not affect the order in $n$, one gets immediately that 
\begin{equation}
a_{ia}^{\dagger} a_{ia} \psi_{kl\Lambda}\simeq \left[ \frac{4}{n^4}\left( \frac{1}{2d(d+1)} A^{\dagger}A +\frac{d+1}{d(d-1)(d+2)} B^{\dagger}B\right)+|\lambda_{\Lambda}| \right] \psi_{kl\Lambda}, ~~ d>1.
\end{equation}
Using the relations (holding at large $n$)
\begin{align}
[S_1,(i^{\dagger}i^{\dagger}j^{\dagger}j^{\dagger})]& \simeq 2n(d+1)(i^{\dagger}i^{\dagger}j^{\dagger}j^{\dagger}), \label{S1} \\
[S_1,(i^{\dagger}j^{\dagger}i^{\dagger}j^{\dagger})]&\simeq 4n (i^{\dagger}i^{\dagger}j^{\dagger}j^{\dagger}),\\
[S_2,(i^{\dagger}i^{\dagger}j^{\dagger}j^{\dagger})]&\simeq4n (i^{\dagger}i^{\dagger}j^{\dagger}j^{\dagger}),\\
[S_2,(i^{\dagger}j^{\dagger}i^{\dagger}j^{\dagger})]&\simeq 4n (i^{\dagger}j^{\dagger}i^{\dagger}j^{\dagger}),\\
[S_3,(i^{\dagger}i^{\dagger}j^{\dagger}j^{\dagger})]&\simeq 2n (i^{\dagger}j^{\dagger}i^{\dagger}j^{\dagger})+2n (i^{\dagger}i^{\dagger}j^{\dagger}j^{\dagger}),\\
[S_3,(i^{\dagger}j^{\dagger}i^{\dagger}j^{\dagger})]&\simeq 4n (i^{\dagger}i^{\dagger}j^{\dagger}j^{\dagger}),
\label{S3}
\end{align}

one can express the action of $S_1, S_2, S_3$ in terms of $A^{\dagger}A,  B^{\dagger}B$ and $ A^{\dagger}B+B^{\dagger}A$ 
\begin{align}
\frac{1}{n}  S_1\psi_{kl\Lambda} &\simeq\left( \frac{1}{d n^4} A^{\dagger}A+G_1(\lambda_{\Lambda})\right) \psi_{kl\Lambda},\label{S1a}\\
\frac{1}{n}  S_2\psi_{kl\Lambda} &\simeq a_{ia}^{\dagger} a_{ia} \psi_{kl\Lambda} \simeq \left[ \frac{4}{n^4}\left( \frac{1}{2d(d+1)} A^{\dagger}A +\frac{d+1}{d(d-1)(d+2)} B^{\dagger}B\right)+|\lambda_{\Lambda}|\right] \psi_{kl\Lambda},\\
\frac{1}{n}  S_3\psi_{kl\Lambda} &\simeq \left[\frac{d+3}{d(d+1)^2n^4} A^{\dagger}A - \frac{2}{d(d+1)n^4}( B^{\dagger}A+ A^{\dagger}B)-\frac{4}{d(d-1)(d+2)n^4}B^{\dagger}B +G_3(\lambda_{\Lambda})\right] \psi_{kl\Lambda}, \label{S3a}
\end{align}
where $G_1(\lambda_{\Lambda}),G_3(\lambda_{\Lambda})$ are the finite dimensional representations of $S_1, S_3$ in the subspaces $P_{\lambda_{\Lambda},k,l} \mathcal{I}_{\tilde{\omega}}$ corresponding to the partition $\Lambda$. Now, by combining (\ref{S1a})-(\ref{S3a}), one can express the action of $S_1+S_2-2S_3$ in terms of $A^{\dagger}A,  B^{\dagger}B, A^{\dagger}B+B^{\dagger}A$ and the finite dimensional representation $G(\lambda_{\Lambda})$ of $\frac{1}{n}( S_1+S_2-2S_3)$ in  $P_{\lambda_{\Lambda},k,l} \mathcal{I}_{\tilde{\omega}}$, which determines $\Omega_1, \Omega_2,  \Omega_3$. \\
In order to get rid of the linear terms in $A+A^{\dagger}$ and $B+B^{\dagger}$ we introduce two new operators 
\begin{align}
\tilde{A}=A+\alpha \mathbb{I},~~~~\tilde{B}=B+\beta \mathbb{I},
\end{align}
and require that the coefficients in front of the A,B-linear terms are zero, which gives exactly (\ref{matrixcondition}) as well as the shift of the vacuum energy $\tilde{e}_0$.

\end{proof}

Note that in the planar limit, $S_2-S_3 \geq 0$, because $S_2 \simeq (\epsilon ij)^{\dagger}(\epsilon ij)$ and $S_{3}\simeq (\epsilon ji)^{\dagger}(\epsilon ij),~~~ (\epsilon ij)= (\epsilon ab)a_{ia}^{\dagger}a_{jb}^{\dagger} $, and thus 
\begin{align}
(\alpha d_{\epsilon ij}^{\dagger}+\beta d_{\epsilon ji}^{\dagger})(\alpha d_{\epsilon ij}+\beta d_{\epsilon ji})= (\alpha^2 + \beta^2)S_2 + 2 \alpha \beta S_3.
\end{align}
By taking $(\alpha^2 + \beta^2)=1$ and $2 \alpha \beta = -1$ we get $S_2-S_3\geq 0$. It turns out that the operator $S_1-S_3$ is not positive definite, i.e. the eigenvalues of $G(\lambda)$ can be negative. Since the operator in Lemma \ref{lemma5} will be used as a lower bound for the full Hamiltonian, one has to assure that (\ref{operator11}) is non-negative definite. This is the case when $\omega \geq g_3$ because asymptotically
$S_1+S_2-2S_3= S_1 +2(S_2-S_3)-S_2 \geq -S_2 \simeq -na_{ia}^{\dagger}a_{ia}=: -n(i^{\dagger}i)$.

We also need further bounds.

\begin{lemma}
\label{lemma6}
We have 
\begin{equation}
\frac{1}{n}S_1 \leq d (i^{\dagger}i) + O(\frac{1}{n}).
\end{equation}
\end{lemma}

\begin{proof}
Consider the following non-negative operator 
\begin{equation}
0 \leq (\alpha (\epsilon k k)\delta_{ij} + \beta (\epsilon j i))^{\dagger}(\alpha (\epsilon k k)\delta_{ij} + \beta (\epsilon j i)),~~ \alpha, \beta \in \mathbb{R}
\end{equation}
where $(\epsilon ji):=Tr(T_{\epsilon} T_a T_b)a_{ja}a_{ib}$. Multiplying the parentheses and using the completeness relation (\ref{completeness relation}) gives asymptotically
\begin{equation}
(d \alpha^2 + 2 \alpha \beta)S_1+\beta^2 S_2\geq 0, ~~ \alpha,\beta \in \mathbb{R}.
\end{equation}
Therefore by taking $\beta= -\frac{\alpha d}{ 2}-\epsilon, ~ \epsilon >0$ and $\alpha > 0  $ we get at large $n$
\begin{equation}
\frac{1}{n}  S_1  \leq \frac{(\frac{\alpha d}{2}+\epsilon)^2}{2 \epsilon \alpha } (i^{\dagger}i).
\end{equation}
The optimal value of the constant is $\min_{\alpha, \epsilon >0} \frac{(\frac{\alpha d}{2}+\epsilon)^2}{2 \epsilon \alpha } = d$.
\end{proof}

Using the last lemma we can prove 
\begin{lemma}
\label{lemma7}
Let $g>0$. Then the operator 
\begin{equation}
H_N=  (i^{\dagger} i) + \frac{g}{n}(T_1+ T_1^{\dagger}+T_2+ T_2^{\dagger}-2(T_3+ T_3^{\dagger}) )
\end{equation}
is non-negative definite at large $n$ for $g < \frac{\sqrt{d}}{2(d-1)}  $.

\end{lemma}

\begin{proof}
Using a similar strategy as in Lemma \ref{lemma6}, we consider another non-negative operator of the form $K_{\epsilon ij}^{\dagger}K_{\epsilon ij}$ with
\begin{align}
K_{\epsilon ij}:=(\alpha( \epsilon kk)\delta_{ij}+  \beta (\epsilon ij)+ \gamma_1 (\epsilon k^{\dagger}k)\delta_{ij}+\gamma_2(\epsilon i^{\dagger}j)+\gamma_3 (\epsilon j^{\dagger}i)).
\end{align}
By computing $K_{\epsilon ij}^{\dagger}K_{\epsilon ij}$ explicitly and using the completeness relation (\ref{completeness relation}) we get at large $n$
\begin{align}
\frac{1}{n}\left[(\alpha  (\gamma_1 d+\gamma_2 + \gamma_3) +\gamma_1 \beta)(T_1+ T_1^{\dagger})+ \beta \gamma_2 (T_2+ T_2^{\dagger})+ \beta \gamma_3 (T_3+ T_3^{\dagger}) \right] \\
\geq \frac{1}{n}\left[-(\alpha^2 d+2 \alpha \beta)S_1- \beta^2 S_2 -(\gamma_1^2 d +2 \gamma_1 \gamma_2 +2 \gamma_1 \gamma_3)(i^{\dagger}i j^{\dagger}j)
 -2 \gamma_2 \gamma_3 (i^{\dagger}j i^{\dagger}j)-(\gamma_2^2+\gamma_3^2)(i^{\dagger}j j^{\dagger}i) \right]\\
 \geq -\left[ \alpha^2 d^2+2 \alpha \beta d+ \beta^2  +\gamma_1^2 d +2\gamma_1 \gamma_2 +2 \gamma_1 \gamma_3+2\gamma_2 \gamma_3
 +d(\gamma_2^2+\gamma_3^2) \right] (i^{\dagger} i) +O(\frac{1}{n}),
 \label{boundt1t3a}
\end{align}
where in the last step we employ the large $n$ relations 
\begin{align}
(i^{\dagger}i j^{\dagger}j)=:(i^{\dagger}i j^{\dagger}j):+n(i^{\dagger} i),\\
(i^{\dagger}j i^{\dagger}j)= :(i^{\dagger}j i^{\dagger}j):+n(i^{\dagger} i),\\
(i^{\dagger}j j^{\dagger}i)= :(i^{\dagger}j j^{\dagger}i): + d n(i^{\dagger} i),
\end{align}
the fact that $0 \leq \frac{1}{n} S_2 \leq  (i^{\dagger} i) $ and Lemma \ref{lemma6}. Constrained minimization of the constant in (\ref{boundt1t3a})
\begin{align}
\min_{\alpha, \gamma_1, \gamma_2, \gamma_3, \beta} \left[ \alpha^2 d^2+2 \alpha \beta d+ \beta^2  +\gamma_1^2 d +2\gamma_1 \gamma_2 +2 \gamma_1 \gamma_3+2\gamma_2 \gamma_3
 +d(\gamma_2^2+\gamma_3^2)  \right]= \frac{2 (d-1)}{\sqrt{d}},\\
 (\alpha (\gamma_1 d+\gamma_2 + \gamma_3) +\gamma_1 \beta)=1 \label{constr1},\\
 \beta \gamma_2=1, ~~~~\beta \gamma_3=-2 \label{constr2},
 \end{align}
gives the bound
\begin{equation}
\frac{1}{n}(T_1+ T_1^{\dagger}+T_2+ T_2^{\dagger}-2(T_3+ T_3^{\dagger}) )\geq -2\frac{d-1}{\sqrt{d}}   (i^{\dagger} i) + O(\frac{1}{n}).
\label{boundt1t3}
\end{equation} As a consequence, $H_N \geq (1-2 g \frac{d-1}{\sqrt{d}} ) (i^{\dagger} i) \geq O(\frac{1}{n})$, provided $g<\frac{\sqrt{d}}{2(d-1)}$.
\end{proof}

 Lemma \ref{lemma7} implies that the operator in (\ref{Toperator}) is non-negative for $\epsilon >\frac{2 \pi^2 (d-1)}{\tilde{\omega}^3 \sqrt{d}}$. The operator $\frac{\pi^2}{n \tilde{\omega}^2}:\left( 
(i^{\dagger} i j^{\dagger} j)+(i^{\dagger} j j^{\dagger} i))-2(i^{\dagger} j i^{\dagger} j)\right):$ is of order $O(\frac{1}{n})$ and thus it does not contribute to the planar limit. Finally by using Lemma \ref{lemma5} with $\epsilon_1=\frac{\pi ^2 (d-1)}{(d+1)  \tilde{\omega}^2},~~\epsilon_2=\frac{2 \pi ^2}{  \tilde{\omega}^2}$, $g=\frac{2 \pi^2}{\tilde{\omega}^2}$ and $\omega= 2 (1 -\epsilon)\tilde{\omega}$ we get the following bound.

\begin{theorem}
The Hamiltonian of the MMM (\ref{quantum membrane}), asymptotically at large $n$, is bounded below by the following non-negative operator
\begin{align}
H_N^{-} :=\Omega_1  \tilde{A}^{\dagger}\tilde{A}+\Omega_2  \tilde{B}^{\dagger} \tilde{B} +\Omega_3( \tilde{B}^{\dagger}\tilde{A}+ \tilde{A}^{\dagger}\tilde{B}) 
+(e_0^{(0)}+\tilde{e}_0) n^2 \mathbb{I} +\sum_{\lambda,\lambda_4=0,k,l}(gG(\lambda) +\omega |\lambda|)P_{\lambda,k,l}+O(\frac{1}{n})\label{theoremlower}
\end{align}
where $\Omega_1:=\Omega_1(\omega,\frac{2 \pi ^2}{\tilde{\omega}^2}),~\Omega_2:=\Omega_2(\omega,\frac{2 \pi ^2}{\tilde{\omega}^2}), \Omega_3:=\Omega_3(\frac{2 \pi ^2}{\tilde{\omega}^2})$, $\tilde{e}_0:=\tilde{e}_0(\omega, \epsilon_1, \epsilon_2, \frac{2 \pi ^2}{\tilde{\omega}^2})$ as well as $\tilde{A}, \tilde{B}$ are defined in Lemma \ref{lemma5}.
\end{theorem}
\begin{proof}

It remains to prove that $H_N^{-} \geq 0$ at large $n$.
Note that $H_N^{-}$ can be rewritten, after performing a Bogoliubov transformation
\begin{align}
\tilde{A}&= (\cos(x) \hat{A} +\sin(x) \hat{B})\sqrt{2d(d+1)},\\
\tilde{B}&=( -\sin(x) \hat{A}+ \cos(x)\hat{B}) \sqrt{\frac{d(d+2)(d-1)}{d+1}} ,
\end{align}
with $[\hat{A},\hat{A}^{\dagger}]=[\hat{B},\hat{B}^{\dagger}]=n^4 \mathbb{I}+O(n^2)$ and $[\hat{B},\hat{A}^{\dagger}]=[\hat{A},\hat{B}^{\dagger}]=O(n^2)$, as a sum of non-negative operators. Indeed, by demanding that the mixed terms $\hat{A}^{\dagger}\hat{B}+\hat{B}^{\dagger}\hat{A}$ vanish, one finds the value of $x$ such that 
\begin{align}
H_N^{-}=\tilde{\Omega}_1  \hat{A}^{\dagger} \hat{A}+\tilde{\Omega}_2 \hat{B}^{\dagger}  \hat{B}
 + \sum_{\lambda,\lambda_4=0,m}(gG(\lambda)  +\omega |\lambda|) P_{\lambda} +(e_0^{(0)}+\tilde{e}_0) n^2 \mathbb{I}+O(\frac{1}{n}),
\end{align}
for some $\tilde{\Omega}_1>0, \tilde{\Omega}_2>0$. Moreover, $g G(\lambda)+\omega|\lambda| \geq 0 ~ \forall \lambda$ since $\omega> g$.
\end{proof}
Since $H_N^{-} \geq (e_0^{(0)}+\tilde{e}_0) n^2 \mathbb{I} + O(\frac{1}{n})$, we get a lower bound for $e_0$ in the planar limit \footnote{In Section 5 we give a perturbative argument, by constructing a perturbative series up to the third order in a certain effective coupling constant, that the planar limit is valid for this model, i.e. the neglected operators of order $O(\frac{1}{n})$ do not affect the ground state energy.}
\begin{align}
e_0^{(lower)}:=e_0^{(0)}+\tilde{e}_0.
\end{align}

We refer the reader to Table \ref{table2} for several numerical values.\\

\section{Perturbative expansion}
\label{perturbations diagrams}
 We have seen that the Hamiltonian of an interacting model with a quartic interaction satisfying the assumption in Lemma 1 can be rewritten as a sum of a non-interacting Hamiltonian $\hat{H}_{0,N}:= H_{0,N}+ N e_0^{(0)} \mathbb{I}$ and a normal ordered interaction (w.r.t. the optimized Fock vacuum $\Psi_0(\tilde{\omega})$)
 \begin{equation}
 H_N=\hat{H}_{0,N} + \frac{1}{N^{\gamma}}:V_N:. 
 \label{perturbed_H}
 \end{equation}
 This suggests a possible perturbative expansion around the spectrum of $\hat{H}_{0,N}$. Denote the eigenvalues and eigenstates  of $\hat{H}_{0,N}$ by $E_{k,N}^{(0)}$ and $\psi_{k,N}^{(0)}$ respectively, where the index $k$ takes values in the appropriate index space $\mathcal{J}_N$. Using the standard technique of stationary perturbations in quantum mechanics (Rayleigh-Schrödinger perturbation theory), i.e. assuming that the actual eigenvalues and eigenstates of the full Hamiltonian, $E_{k,N}$ and $\psi_{k,N}$, can be expanded in a power series in $\epsilon_N:= \epsilon N^{-\gamma}$, where $\epsilon$ is a book keeping parameter \footnote{In fact $\epsilon$ is not only a book keeping parameter but also a manifestation of the silent assumption about the existence of a small parameter hidden in the model, which in general does not have to exist. However, we will see that for both matrix models such a parameter exists.}, one can compute a first few terms of the perturbative expansion for finite $N$, take the limit $N \rightarrow \infty$, and in the end put $\epsilon=1$. The corrections to the eigenvalue $E_{k,N}$ of $H_N$ up to the third order read
 \begin{align}
 E_{k,N}&= E_{k,N}^{(0)}+\epsilon_N E_{k,N}^{(1)} + \epsilon_N^2 E_{k,N}^{(2)}+\epsilon_N^3 E_{k,N}^{(3)}+ O(\epsilon^4),\\
 E_{k,N}^{(1)}&=\langle \psi_k^{(0)},:V_N: \psi_k^{(0)} \rangle, \label{perturbative expansion}\\
 E_{k,N}^{(2)}&=\langle \psi_k^{(0)},:V_N: \frac{Q_0}{E_{k,N}^{(0)}-\hat{H}_{0,N} }:V_N: \psi_k^{(0)}\rangle,\\ 
 E_{k,N}^{(3)}&=\langle \psi_k^{(0)},:V_N: \frac{Q_0}{E_{k,N}^{(0)}-\hat{H}_{0,N} }:V_N:\frac{Q_0}{E_{k,N}^{(0)}-\hat{H}_{0,N} } :V_N: \psi_k^{(0)}\rangle \\&-E_{k,N}^{(1)}\langle \psi_k^{(0)},:V_N: \frac{Q_0}{(E_{k,N}^{(0)}-\hat{H}_{0,N} )^2}:V_N: \psi_k^{(0)}\rangle,
 \end{align}
 where $Q_0$ is the orthogonal projection on $span(\psi_k^{(0)})^{\perp}$ and $\frac{Q_0}{E_{k,N}^{(0)}-\hat{H}_{0,N} }:=Q_0(E_{k,N}^{(0)}-\hat{H}_{0,N})^{-1}Q_0$.
It is of course not clear whether the series converges, not even if separate terms do in the limit $n \rightarrow \infty$ and one usually expects that such a series is only asymptotic. However we shall see below that the first three terms turn out to provide an astonishingly accurate approximation for the ground state energy and the spectral gap for the AO and to be consistent with the bounds obtained for the MMM. We will also show that the spectral gap for the MMM remains finite at large $n$ at least up to the $2^{nd}$ order. 

Another important aspect of the considered here perturbative expansion is that it establishes a natural connection between the calculation presented in the previous sections and the planar limit in gauge field theories. In the next subsection we will give examples of planar and non-planar contributions (resp. leading and subleading) to the vacuum energy using a diagramatic representation of certain terms popping up in the above perturbative series (\ref{perturbative expansion}), which also justifies why one can neglect the operators of order $O(\frac{1}{n})$ in (\ref{Hneps3}) and (\ref{lastpart}), at least in the vacuum energy calculations.
 
\subsection{Vacuum energy corrections}
Let us remind that the $0^{th} $ term gives the gaussian variational upper bound and point out that the $1^{st}$ order term is always zero. We will see that (at least up to $3^{rd}$ order) the leading terms of the perturbative contributions to the vacuum energy for the models considered in this note correspond to \emph{connected} planar diagrams in the diagramatic representation and they are proportional to $n^2$.
Let us compute the second and third order corrections to the ground state energy of the AO (\ref{1matrix}). We put $E_{0,N}=N e_0$ and $E_{0,N}^{(0)}=N e_0^{(0)}$
\begin{itemize}
\item $2^{nd}$ order correction for the AO
\begin{align}
\frac{e_{0,N}^{(2)}}{2}=\lim_{n \rightarrow \infty}\frac{\epsilon_N^2}{N}\langle  \psi_0,:V_N: \frac{Q_0}{E_{0,N}^{(0)}-\hat{H}_{0,N} }:V_N: \psi_0 \rangle=\\-\lim_{n \rightarrow \infty} \frac{g^2 \epsilon^2}{ 64 n^4 \tilde{\omega}^5} \langle \psi_0,A A^{\dagger} \psi_0 \rangle
 = -\frac{g^2 \epsilon^2}{ 16 \tilde{\omega}^5},
 \label{2ndorder}
\end{align}
where we have used that $\hat{H}_{0,N} A^{\dagger }\psi_0= (4 \tilde{\omega} +N e_0^{(0)})A^{\dagger }\psi_0$. 

Let us point out the connection between our perturbative expansion and the topological expansion with the leading, genus 0 order, being the planar limit. The contraction occuring in (\ref{2ndorder}) has the form $Tr(T_a T_b T_c T_d) Tr(T_{\tilde{a}} T_{\tilde{b}} T_{\tilde{c}} T_{\tilde{d}}) \delta_{a \tilde{b}}\delta_{b \tilde{a}} \delta_{d \tilde{c}} \delta_{c \tilde{d}}$ and it can be represented graphically as
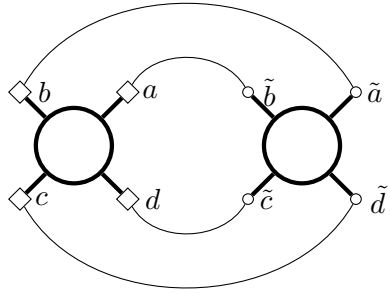
\begin{center}
\begin{tikzpicture}
    \tikzstyle{every node}=[draw,circle,fill=white,minimum size=4pt,
                            inner sep=0pt]

    % First, draw the inner hexagon with a ``pin'' -- namely, (3214)
    \draw[ultra thick] (0,0) node[circle, minimum height=1cm,minimum width=1cm,draw] (1)  {};
    \draw (1/1.41,1/1.41) node[diamond, minimum height=0.3cm,minimum width=0.3cm,draw] (2) [label=right:$a$]   {};
    \draw  (-1/1.41,1/1.41) node[diamond, minimum height=0.3cm,minimum width=0.3cm,draw] (3) [label=right:$b$]  {};
    \draw (-1/1.41,-1/1.41) node[diamond, minimum height=0.3cm,minimum width=0.3cm,draw] (4) [label=right:$c$]  {};
    \draw (+1/1.41,-1/1.41) node[diamond, minimum height=0.3cm,minimum width=0.3cm,draw] (5)  [label=right:$d$] {};
            
             \draw[ultra thick] (1) -- (2);
              \draw[ultra thick] (1) -- (3);
               \draw[ultra thick] (1) -- (4);
                \draw[ultra thick] (1) -- (5);
     
     \draw[ultra thick] (3,0) node[circle, minimum height=1cm,minimum width=1cm,draw] (6)  {};
    \draw (3+1/1.41,+1/1.41) node (7) [label=right:$\tilde{a}$] {};
    \draw (3-1/1.41,+1/1.41) node (8) [label=right:$\tilde{b}$]  {};
    \draw (3-1/1.41,-1/1.41) node (9) [label=right:$\tilde{c}$]  {};
    \draw (3+1/1.41,-1/1.41) node (10)  [label=right:$\tilde{d}$] {};
            
             \draw[ultra thick] (6) -- (7);
              \draw[ultra thick] (6) -- (8);
               \draw[ultra thick] (6) -- (9);
                \draw[ultra thick] (6) -- (10);

\draw (3) to [out=60,in=120] (7);
\draw (4) to [out=-60,in=-120] (10);
\draw (2) to [out=60,in=120] (8);
\draw (5) to [out=-60,in=-120] (9);
    
\end{tikzpicture}
\end{center}

where a ring with 4 outgoing branches represents the trace of a product of 4 matrices $T_a$, a circle at the end represents an annihilation operator while a square stands for a creation operator carrying the indicated index. Thin lines between vertices represent Kronecker deltas resp. Wick contractions between annihilation and creation operators. It is clear that the maximal power of $n$ is attained when adjacent indices in the two traces are contracted, which corresponds exactly to the drawn above planar contraction. If one contracts e.g. $b $ with $\tilde{b}$ and $a$ with $\tilde{a}$ instead then the contribution to (\ref{2ndorder}) is subleading and the resulting diagram cannot be drawn in the plane without intersections of the contraction lines. 
\item $3^{rd}$ order correction for the AO
\begin{align}
\frac{e_{0,N}^{(3)}}{2}=\lim_{n \rightarrow \infty} \frac{\epsilon_N^3}{N} \langle \psi_0,:V_N: \frac{Q_0}{E_{0,N}^{(0)}-\hat{H}_{0,N} }:V_N:\frac{Q_0}{E_{0,N}^{(0)}-\hat{H}_{0,N} } :V_N: \psi_0\rangle  \\
= \lim_{n \rightarrow \infty} \frac{g^3 \epsilon^3}{ 4^5 n^5 \tilde{\omega}^8} \langle \psi_0,A (a^{\dagger}a^{\dagger} a a) A^{\dagger}\psi_0 \rangle  
 = \frac{g^3 \epsilon^3}{ 4^2 \tilde{\omega}^8}.
 \label{3rdorder}
\end{align}

A graphical representation of this contribution has the following form 
    
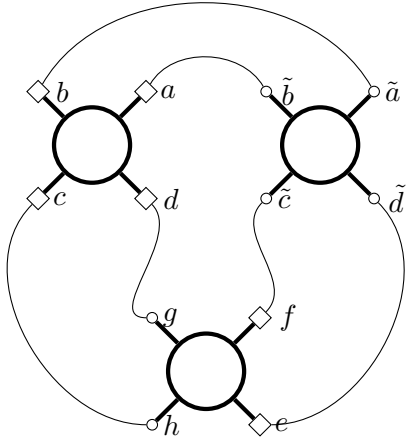
\begin{center}
\begin{tikzpicture}
    \tikzstyle{every node}=[draw,circle,fill=white,minimum size=4pt,
                            inner sep=0pt]

    % First, draw the inner hexagon with a ``pin'' -- namely, (3214)
    \draw[ultra thick] (0,0) node[circle, minimum height=1cm,minimum width=1cm,draw] (1)  {};
    \draw (1/1.41,1/1.41) node[diamond, minimum height=0.3cm,minimum width=0.3cm,draw] (2) [label=right:$a$]   {};
    \draw  (-1/1.41,1/1.41) node[diamond, minimum height=0.3cm,minimum width=0.3cm,draw] (3) [label=right:$b$]  {};
    \draw (-1/1.41,-1/1.41) node[diamond, minimum height=0.3cm,minimum width=0.3cm,draw] (4) [label=right:$c$]  {};
    \draw (+1/1.41,-1/1.41) node[diamond, minimum height=0.3cm,minimum width=0.3cm,draw] (5)  [label=right:$d$] {};
            
             \draw[ultra thick] (1) -- (2);
              \draw[ultra thick] (1) -- (3);
               \draw[ultra thick] (1) -- (4);
                \draw[ultra thick] (1) -- (5);
     
     \draw[ultra thick] (3,0) node[circle, minimum height=1cm,minimum width=1cm,draw] (6)  {};
    \draw (3+1/1.41,+1/1.41) node (7) [label=right:$\tilde{a}$] {};
    \draw (3-1/1.41,+1/1.41) node (8) [label=right:$\tilde{b}$]  {};
    \draw (3-1/1.41,-1/1.41) node (9) [label=right:$\tilde{c}$]  {};
    \draw (3+1/1.41,-1/1.41) node (10)  [label=right:$\tilde{d}$] {};
            
             \draw[ultra thick] (6) -- (7);
              \draw[ultra thick] (6) -- (8);
               \draw[ultra thick] (6) -- (9);
                \draw[ultra thick] (6) -- (10);

    \draw[ultra thick] (1.5,-3) node[circle, minimum height=1cm,minimum width=1cm,draw] (11)  {};
    \draw (1.5+1/1.41,-1/1.41-3) node[diamond, minimum height=0.3cm,minimum width=0.3cm,draw] (12) [label=right:$e$]   {};
    \draw  (1.5+1/1.41,1/1.41-3) node[diamond, minimum height=0.3cm,minimum width=0.3cm,draw] (13) [label=right:$f$]  {};
    \draw (1.5+-1/1.41,1/1.41-3) node (14) [label=right:$g$]  {};
    \draw (1.5-1/1.41,-1/1.41-3) node (15)  [label=right:$h$] {};
            
             \draw[ultra thick] (11) -- (12);
              \draw[ultra thick] (11) -- (13);
               \draw[ultra thick] (11) -- (14);
                \draw[ultra thick] (11) -- (15);
                
                \draw (3) to [out=60,in=120] (7);
\draw (9) to [out=225,in=45] (13);
\draw (2) to [out=60,in=120] (8);
\draw (10) to [out=-45,in=0] (12);
\draw (5) to [out=-45,in=180] (14);
\draw (4) to [out=225,in=180] (15);

\end{tikzpicture}
\end{center}

where the bottom circle represents the operator $(a^{\dagger} a^{\dagger}a a)=Tr(T_e T_f T_g T_h) a^{\dagger}_e a^{\dagger}_f a_g a_h $. If one replaces it with the operator $:(a^{\dagger}a  a^{\dagger} a): Tr(T_e T_f T_g T_h) a^{\dagger}_e a^{\dagger}_g a_f a_h$ then the resulting diagram is no longer planar and it leads to a subleading contribution to (\ref{3rdorder}).

These corrections decrease the error of determining the ground state energy at large $g$ to  $1.5 \permil $ and $1.1 \permil$  resp. (see Table \ref{table3}).
\begin{table}[H]
\caption{Comparison of the exact ground state energies $e_0$ \cite{planar} with the optimized Fock space ground state energies up to $3^{rd}$ order of the perturbation expansion $e_0^{(0)}/2,e_0^{(2)}/2, e_0^{(3)}/2$ for several values of the coupling constant $g$ for the AO.} 
\label{table3}
\begin{center}
    \begin{tabular}{ | l | l | l |l |l |l | p{5cm} |}
    \hline
g & $e_0^{(0)}/2$& $e_0^{(2)}/2$ & $e_0^{(3)}/2$ &  $e_0$  \\ \hline
 0.01 &0.505 &0.505& 0.505 & 0.505  \\ \hline
 0.1 & 0.543&0.542& 0.542&0.542 \\        \hline
0.5 &0.653& 0.651& 0.651&  0.651  \\ \hline
1.0 &  0.743& 0.740& 0.740 &  0.740  \\  \hline
50 & 2.235&2.214& 2.219& 2.217  \\        \hline
1000 & 5.968&5.907& 5.922&5.915   \\        \hline
$g \rightarrow \infty $  & 0.59527 $g^{\frac{1}{3}}$ & 0.589075 $g^{\frac{1}{3}}$ &  0.59062 $g^{\frac{1}{3}}$&  0.58993 $g^{\frac{1}{3}}$   \\        \hline
    \end{tabular}
    \end{center}
    \end{table}

For the MMM one can easily repeat the calculation
\item $2^{nd}$ order correction for the MMM

\begin{align}
e_{0,N}^{(2)}=-\lim_{n \rightarrow \infty}\frac{1}{n^4}\frac{\pi^4}{8 \tilde{\omega}^5}\langle \psi_0,\left((iijj)-(ijij) \right)\left( (k^{\dagger}k^{\dagger}l^{\dagger}l^{\dagger})-(k^{\dagger}l^{\dagger}k^{\dagger}l^{\dagger})\right) \psi_0 \rangle =-\frac{6d(d-1)\pi^4}{8 \tilde{\omega}^5}
\end{align}

\item $3^{rd}$ order correction for the MMM
\begin{align}
e_{0,N}^{(3)}= \lim_{n \rightarrow \infty} \frac{1}{n^5}\frac{\pi^6}{32\tilde{\omega}^8}  \langle \psi_0,\lbrace (iijj)-(ijij)\rbrace \lbrace S_1+S_2 -2 S_3 \rbrace \lbrace (i^{\dagger}i^{\dagger}j^{\dagger}j^{\dagger})-(i^{\dagger}j^{\dagger}i^{\dagger}j^{\dagger}) \rbrace\psi_0 \rangle \\ = d \left(d^2+10 d-11\right) \frac{\pi^6}{8 \tilde{\omega}^8},
\end{align}
where we have used the asymptotic formulas (\ref{S1})-(\ref{S3}). A diagramatic representation of these contributions is qualitatively the same as for the AO, which leads to the conclusion that the operator $:((i^{\dagger}i j^{\dagger}j)+(i^{\dagger}j j^{\dagger} i)-2(i^{\dagger}j i^{\dagger} j)):$ in (\ref{lastpart}) is negligible in the planar limit. 

\end{itemize}

The reader is referred to Table \ref{table2} for several numerical values, which show that the proposed perturbative expansion is consistent with our upper and lower bounds. One can see that the operators $Tr(a^{\dagger4}), Tr(a^4)$, i.e. $A,A^{\dagger}$ for the AO and $(iijj), (ijij),(i^{\dagger} i^{\dagger}j^{\dagger}j^{\dagger}), (i^{\dagger}j^{\dagger}i^{\dagger}j^{\dagger})$ for the MMM are the main sources of corrections to the vacuum energy for both models. 
\begin{table}[H]

\caption{Comparison of the ground state energies for the MMM in the $0^{th}$, $2^{nd}$ and $3^{rd}$ order as well as the lower bound $e_0^{(lower)}$ found in Section \ref{spectral bounds MMM}} 
\label{table2}
\begin{center}
    \begin{tabular}{ | l | l | l |l |l | l | p{5cm} |}
    \hline
    d       &     3    &  9   &    15 & 25 & 35 \\ \hline
$e_0^{(0)}$ &  9.653 &  45.968 & 92.324 & 184.158 & 289.562 \\ \hline
$e_0^{(2)}$ &  9.351 &  45.609 & 91.912 & 183.679 & 289.030  \\        \hline
$e_0^{(3)}$ &  9.439 & 45.646 & 91.944  &183.709 & 289.060\\ \hline
$e_0^{(lower)}$ & 9.349 &45.583 &91.887& 183.658 &289.013  \\ \hline
    \end{tabular}
    \end{center}
    \end{table}

\subsection{Spectral gap corrections}

The spectrum of the singlet sector for the AO consists of a divergent vacuum energy and an infinite family of equally spaced excited states i.e. with a constant, finite energy gap $\omega(g)$ \cite{shapiro,Yaffe,planar limit marchesini onofri}:
\begin{equation}
E_{\lambda}= e_0 n^2+ \omega(g) \sum_{j} j \lambda_j.
\end{equation}
The degeneracy of an energy level $E_k=n^2 e_0 + k \omega(g)$ is given by the number of partitions of $k$. This result is restored in our perturbative expansion, even in the $0^{th}$ order, so for $U(n)$ invariant eigenstates of $\hat{H}_{0,N}= \tilde{\omega}(a^{\dagger}a)+ \frac{e_0^{(0)}}{2} n^2 \mathbb{I}$.

For later convenience we choose the second excited eigenstate of $\hat{H}_{0,N}$, i.e. $\psi_{\lambda}= \mathcal{N}_{\lambda} (a^{\dagger} a^{\dagger})\Psi_0(\tilde{\omega})$, $\lambda_i=0, ~i\neq 2,~\lambda_2=1$, as a starting point for the perturbative expansion.
The first order correction is of order $O(1)$
\begin{align}
E_{\lambda}^{(1)}= \lim_{n \rightarrow \infty} \frac{g}{\tilde{\omega}^2n} \langle  \psi_{\lambda},(a^{\dagger}a^{\dagger} a a)\psi_{\lambda}\rangle =\lim_{n \rightarrow \infty} \mathcal{N}_{\lambda}^2 \frac{g}{\tilde{\omega}^2n} \langle \psi_{\lambda}, (aa)(a^{\dagger}a^{\dagger} a a) (a^{\dagger}a^{\dagger}) \psi_{\lambda}\rangle =\frac{g}{\tilde{\omega}^2}.
\end{align}
The second and third order contributions $E_{\lambda}^{(2)},E_{\lambda}^{(3)}$ obtained in this way diverge like $n^2$. However, the divergent parts are exactly equal to the corresponding corrections to the vacuum energy and thus the renormalized energies are finite. 
\begin{align}
& E_{\lambda,R}^{(2)}:=\lim_{n \rightarrow \infty}(E_{\lambda}^{(2)}-e_0^{(2)}n^2) \label{canc1} \\ &= \lim_{n \rightarrow \infty}\left[ \frac{g^2}{4^2\tilde{\omega}^4 n^2} \left(16 \langle \psi_{\lambda},(a^{\dagger}a a a) \frac{Q_0}{E_{\lambda,N}^{(0)}-\hat{H}_{0,N}}(a^{\dagger}a^{\dagger} a^{\dagger} a)\psi_{\lambda}\rangle  + \langle \psi_{\lambda},A \frac{Q_0}{E_{\lambda,N}^{(0)}-\hat{H}_{0,N} }A^{\dagger}\psi_{\lambda}\rangle  \right) -e_0^{(2)}n^2\right] \\ & =-\frac{ 5 g^2}{ \tilde{\omega}^5},
\end{align}
and similarly the $3^{rd}$ order renormalized term
\begin{align}
E_{\lambda,R}^{(3)}:=\lim_{n \rightarrow \infty}(E_{\lambda}^{(3)}-e_0^{(3)}n^2) = \frac{310 g^3}{32\tilde{\omega}^8}.
\label{canc2}
\end{align}

Therefore the spectral gap $\omega(g)$ up to the third order becomes
\begin{equation}
\omega(g)=\frac{1}{2}(2 \tilde{\omega}+E_{\lambda,R}^{(1)}+E_{\lambda,R}^{(2)}+E_{\lambda,R}^{(3)})=\tilde{\omega}+\frac{g}{\tilde{\omega}^2}-\frac{5 g^2}{2 \tilde{\omega}^5}+\frac{155 g^3}{32 \tilde{\omega}^8}.
\end{equation}

Table \ref{table4} shows several numerical values. The accuracy of our approximation is very good, also in the strong coupling limit $g \rightarrow \infty$. The reason why it works so well is the following. In the optimized Fock space formulation the Hamiltonian takes the form, cp. (\ref{Hneps1})-(\ref{Hneps3}),
\begin{align}
H_N +\beta_N \mathbb{I}= \tilde{\omega}\left[(a^{\dagger}a) + \frac{g}{4\tilde{\omega}^3} (\text{interaction terms}) \right],
\end{align}
hence the effective coupling constant is not $g$, but $\tilde{g}:=\frac{g}{4\tilde{\omega}^3}$. By using eq.(\ref{opt freq}), one gets $\max_g \tilde{g}=\frac{1}{16}$, attained in the limit $g \rightarrow \infty$, which shows that the system in the optimized Fock space is coupled weakly even if the coupling $g$ in the original formulation is strong. Let us also point out that each order in $\tilde{g}$ gives the correct scaling of the spectrum as 
\begin{align}
\tilde{\omega}\tilde{g}^k \propto \frac{g^k}{\tilde{\omega}^{3k-1}} \propto g^{\frac{1}{3}}, ~~~ g\rightarrow \infty.
\end{align}

\begin{table}[H]
\caption{Comparison of the exact value of the spectral gap $\omega(g)$ \cite{singlet spectrum mondello onofri} with the perturbative expansion in the optimized Fock space $\omega^{(0)}=\tilde{\omega}$, $\omega^{(1)}$, $\omega^{(2)}$, $\omega^{(3)}$.}

\label{table4}
\begin{center}
    \begin{tabular}{ | l | l | l |l |l |l | p{5cm} |}
    \hline
g & $\omega(g)$& $\omega^{(0)}$ & $\omega^{(1)}$ & $\omega^{(2)}$ & $\omega^{(3)}$ \\ \hline
 2 &2.454 &2.166& 2.592 & 2.382&2.463 \\ \hline
 50 & 6.811 &5.905&7.340& 6.468&6.878 \\        \hline
200 &10.76& 9.319& 11.62& 10.20&10.88  \\ \hline
1000 & 18.37&15.90&19.85 & 17.39&18.58  \\  \hline
    \end{tabular}
    \end{center}
    \end{table}
    
Let us repeat the calculation for the MMM. As the starting point for the perturbation expansion we choose an analogue of the previously considered for the AO $\psi_{\lambda}$, i.e. the first $SO(d)\times SU(n)$ invariant excited state of $\hat{H}_{0,N}= 2 \tilde{\omega}(i^{\dagger} i) + e_0^{(0)}n^2 \mathbb{I}$, namely $\psi_{\Lambda}=\mathcal{N}_{\Lambda} (i^{\dagger}i^{\dagger})\Psi_0(\tilde{\omega}) $ and thus $E_{\Lambda}^{(0)}=4 \tilde{\omega}+e_0^{(0)}n^2$. The first order correction reads

\begin{align}
E_{\Lambda}^{(1)}=\lim_{n \rightarrow \infty}\frac{2 \pi^2}{\tilde{\omega}^2 n} \langle  \psi_{\Lambda},(S_1+S_2-2S_3)\psi_{\Lambda}\rangle =\frac{4 \pi^2}{\tilde{\omega}^2}(d-1),
\end{align}
and the second order is again divergent with the singular part equal to the corresponding correction to the vacuum energy $e_0^{(2)}n^2$. Thus the renormalized contribution

\begin{align}
E_{\Lambda,R}^{(2)}&:=\lim_{n \rightarrow \infty}(E_{\Lambda}^{(2)} - e_0^{(2)}n^2) \label{canc3}= \lim_{n \rightarrow \infty}[- \frac{1}{4\tilde{\omega}} \frac{4 \pi^4}{\tilde{\omega}^4 n^2} \langle \psi_{\Lambda},(T_1+T_2-2T_3)(T_1^{\dagger}+T_2^{\dagger}-2T_3^{\dagger}),\psi_{\Lambda}\rangle  \\&-\frac{1}{8 \tilde{\omega}}\frac{\pi^4}{\tilde{\omega}^4 n^2}  \langle \psi_{\Lambda},((iijj)-(ijij))((k^{\dagger} k^{\dagger} l^{\dagger} l^{\dagger})-(k^{\dagger}l^{\dagger}k^{\dagger}l^{\dagger})),\psi_{\Lambda}\rangle - e_0^{(2)}n^2] \\&= -\frac{40 \pi^4}{\tilde{\omega}^5}(d-1)-\frac{\pi^4}{\tilde{\omega}^5}(d^2+7),
\end{align}
is finite at latge $n$.
Table \ref{table5} contains several numerical values of the renormalized energies $E_{\Lambda,R}^{(k)}=\sum_{i=0}^{k}(E_{\Lambda}^{(i)}-e_0^{(i)}n^2),~k=1,2$. It is apparent that the perturbative series has better convergence properties for higher dimensions $d$, which can be justified by localizing an expansion parameter in the Hamiltonian, cp. (\ref{ABSoperator})-(\ref{lastpart}), 
\begin{align}
H_N +\beta_N \mathbb{I}= 2\tilde{\omega}\left[(i^{\dagger}i) + \frac{\pi^2}{\tilde{\omega}^3} (\text{interaction terms}) \right],
\end{align}
hence the effective coupling constant is $\tilde{g}:=\frac{\pi^2}{\tilde{\omega}^3}=\frac{1}{4(d-1)}$.

\begin{table}[H]

\caption{The perturbative expansion for the renormalized energy (the vacuum energy subtracted) of the first $SO(d)\times SU(n)$ invariant excited state for the MMM at large $n$}

\label{table5}
\begin{center}
    \begin{tabular}{ | l | l | l |l |l |l | p{5cm} |}
    \hline
d &3 & 9 &15& 25& 35 \\ \hline
 $E_{\Lambda,R}^{(0)}$ &17.16 &27.24& 32.82& 39.29& 44.12  \\ \hline
 $E_{\Lambda,R}^{(1)}$ & 21.45 &34.05&41.03&49.11&55.15 \\        \hline
$E_{\Lambda,R}^{(2)}$ &16.09& 31.92& 39.57& 48.09&54.34 \\ \hline
    \end{tabular}
    \end{center}
    \end{table}

\section{Concluding remarks}
We have seen in Section \ref{opt Fock space} that the special choice of the Fock vacuum (resp. Fock space frequencies, $\tilde{\omega}_I$) eliminates one of the divergent parts of the Hamiltonian, i.e. $Tr(a^2+a^{\dagger 2}) \propto O(n)$ and minimizes the vacuum expectation value $\langle \Psi_0(\omega), H_N \Psi_0(\omega) \rangle=e_0^{(0)} n^2$ providing a rigorous upper bound for the ground state energy at large $n$ in a concise way for a rather general family of Hamiltonians with a quartic potential (\ref{general}).
A further study of the remaining divergent operators, i.e. $\frac{1}{n} Tr(a^{\dagger 4}+ a^4)$ allows to produce astonishingly good lower bounds for the ground state energy for the AO as well as for the MMM (Theorems 1 and 2) in the planar limit. A similar mechanism shall lead to the true vacuum in Fock space, however the task is rather challenging. Our results suggest that one should consider bases consisting of non-linear coherent states instead of polynomials in simple $U(n)/SU(n)$ invariants (\ref{partitions basis}), (\ref{partitions basis2}), which would fully capture the divergent behaviour of $\frac{1}{n} Tr(a^{\dagger 4}+ a^4)$. For instance, a possible vacuum state $\phi_0$ for the operator $\frac{1}{n^4} B^{\dagger}B$ from Lemma \ref{lemma2} has the following form
\begin{equation}
\phi_0:= \mathcal{N}e^{\tilde{\gamma} C^{\dagger}} \psi_0,~~ [A,C^{\dagger}]=\mathbb{I}, ~~\tilde{\gamma}=-\frac{\beta n^3}{\alpha+  \beta \gamma},
\label{renormalized vacuum}
\end{equation}
and $C$ can be found recursively as an infinite sum of single-trace operators with increasing length, $C=\sum_{k=0}C_k$, $C_0= \frac{1}{4n^4}A$.

Nevertheless, the two considered here matrix models admit a perturbative formulation, which provides evidence for the validity of the planar limit and it seems to be efficient in approximating the spectrum, including the strong-coupling regime for the AO, allowing to rederive the well-knwon result of Brezin et al \cite{planar} for the ground state energy with high accuracy. In particular, the Hamiltonian for the Membrane Matrix Models (\ref{classmembr}), originally not containing a quadratic term, gains an effective mass in the optimized Fock space and can be treated as a perturbation of a non-interacting Hamiltonian by the original normal ordered quartic potential with the effective coupling constant $\tilde{g}= \frac{1}{4(d-1)}$, which means that the model contains a small hidden parameter. This is perhaps not surprising in the context of the established validity of $1/D$ expansions (where $D=d+1$ is the dimension of space-time) in various Yang-Mills theories, see \cite{1 over d expansion, 1 over d expansion 2}. Our numerical results, Table \ref{table2} and \ref{table5}, show that the perturbative expansion is stable, especially for higher dimensions of the embedding space, the results are more accurate for larger values of $d$, and that it leads to a finite energy gap at least up to the second order. One expects that cancellations of the vacuum contributions (corresponding to vacuum diagrams), like in formulas (\ref{canc1}), (\ref{canc2}) and (\ref{canc3}), should occur in any order, depending on the validity of the linked cluster theorem for this model. Therefore one can conjecture that the spectrum of the properly rescaled family of Membrane Matrix Models (\ref{quantum membrane}) in the large $n$ limit remains purely discrete with a finite spectral gap and a divergent vacuum energy, i.e. a generic energy level $E_{\Lambda}$ can be written as
\begin{equation}
E_{\Lambda}= e_0 n^2 + e_{\Lambda},~~ 0< e_{\Lambda} < \infty. 
\end{equation}

\section*{Acknowledgements}

The author would like to thank J. Hoppe, G. Lambert, E. Langmann, D. Lundholm and O.T. Turgut for valuable discussions, T. Morita and E. Onofri  for e-mail correspondence as well as KTH and the Swedish Research Council for financial support.

\appendix

\section{$SU(2)$ Matrix Models}

We will demonstrate that our method is useful also for finite $N$ by comparing the spectrum computed in the optimized Fock space with the existing results. Consider the $n=2$ Matrix Model
\begin{equation}
H_{d}=p_{ia} p_{ia} + \frac{1}{2} f^{(2)}_{abc}f^{(2)}_{ab'c'}x_{ib}x_{ib'}x_{jc}x_{jc'},
\label{n=2 model}
\end{equation}
Using the Pauli matrices $\sigma_a,~a=1,2,3$, normalized as follows 
\begin{align}
Tr(\sigma_a \sigma_b)= \delta_{ab},~~~\left[\sigma_a, \sigma_b \right]= i \sqrt{2} \sigma_c,
\end{align}
as a representation of the $su(2)$ algebra, we get $ f^{(2)}_{abc}=\frac{2 \pi\cdot 2^{\frac{3}{2}}}{i}Tr(\sigma_a [\sigma_b, \sigma_c])=8 \pi \epsilon_{abc}$. Therefore 

\begin{equation}
H_{d}=p_{ia} p_{ia} + 32 \pi^2 \epsilon_{abc} \epsilon_{ab'c'}x_{ib}x_{ib'}x_{jc}x_{jc'},
\label{n=2 model2}
\end{equation}
In order to compare our computation with the spectrum of the "standard" $SU(2)$ model

\begin{equation}
\tilde{H}_{d}=p_{ia} p_{ia} + \frac{1}{2} \epsilon_{abc} \epsilon_{ab'c'}x_{ib}x_{ib'}x_{jc}x_{jc'},
\label{n=2 model standard}
\end{equation}
we need to employ the formulas (\ref{MMM tensors}) with the interaction part divided by the factor $64 \pi^2$ and put $\gamma_{crit}=0$, i.e.
\begin{align}
\tilde{f}(2)= \frac{f(2)}{64 \pi^2}=\frac{3 d(d-1)}{4 \omega^2},\\
\tilde{A}_{IJ}^{(2 \pm)}= \frac{A_{IJ}^{(2 \pm)}}{64 \pi^2}=\left(\frac{d-1}{\omega^2} \pm \omega \right)\delta_{IJ},\\
-\tilde{\beta} = \frac{3d \omega}{2}+\frac{3d(d-1)}{4\omega^2}.
\label{su2 tensors}
\end{align}
Therefore the optimized frequency $\tilde{\omega}$ solving $\frac{\partial \tilde{\beta}}{\partial \omega}$, resp. $\tilde{A}_{IJ}^{(2 -)}=0$ and the variational upper bound become
\begin{align}
\tilde{\omega}=(d-1)^{\frac{1}{3}},\\
e_0^{(0)}=-\tilde{\beta}(\tilde{\omega})=\frac{9}{4}d(d-1)^{\frac{1}{3}}.
\end{align}
 The optimized Fock space decomposition for (\ref{n=2 model standard}) takes the form
 \begin{equation}
\tilde{H}_{d} = 2 \tilde{\omega}a_{ia}^{\dagger}a_{ia} + e_0^{(0)}\mathbb{I}  +:V:,
\label{n=2 model standard opt}
\end{equation}
One can perform a perturbative expansion according to (\ref{perturbative expansion}), treating  $:V:=\frac{1}{2} \epsilon_{abc} \epsilon_{ab'c'}:x_{ib}x_{ib'}x_{jc}x_{jc'}:$ as a perturbation. Table $\ref{table6}$ contains a comparison of the exact ground state energies for $d=2$, \cite{variational orthogonalization}, and $d=3$, \cite{munster}, with our computation.
\begin{table}[H]

\caption{Comparison of the ground state energy of the $d=2,3$ $SU(2)$ Matrix Model obtained perturbatively with the "exact" values $e_0^{exact}$ (obtained by the Rayleigh-Ritz method; note the rescaling by 2) with our perturbative expansion in the optimized Fock space up to $2^{nd}$ order }
\label{table6}
\begin{center}
    \begin{tabular}{ | l | l |l |  p{5cm} |}
    \hline
d &2 & 3 \\ \hline
 
 $e_0^{(0)}$ & 4.500 &8.504 \\        \hline
$e_0^{(2)}$ &4.219& 8.238\\  \hline
$e_0^{exact}$ &4.230 &8.233 \\ \hline
    \end{tabular}
    \end{center}
    \end{table}

    \section{Trivial instructive example}
    
    Let us demonstrate the technique of getting a lower bound for the ground state energy applied in Section \ref{spectral bounds MMM} and \ref{spectral bounds AO} on the following trivial example
    
    \begin{equation}
    H_N=\sum_{a=1}^{N}(\omega a^{\dagger}_a a_a+ \gamma (a_a a_a + a^{\dagger}_a a^{\dagger}_a)),
    \end{equation}
acting on the Hilbert space of $su(n)$ ($N=n^2-1$) invariant wave functions embedded naturally, as described in Section 2, in the infinite tensor product of $L^2$ spaces. Performing a Bogoliubov transformation from $a_a^{\dagger}, a_a$ to a new pair of canonical creation-annihilation operators $b_a^{\dagger},b_a$
\begin{align}
a_a^{\dagger}=\alpha b_a + \beta b_a^{\dagger},\\
a_a=\alpha b_a^{\dagger} + \beta b_a,
\end{align}
with
\begin{align}
\alpha&=\sin x\\
\beta&= \cos x,~~~~ x \in \mathbb{R},
\end{align}
allows to rewrite $H_N$ in the standard form 
\begin{equation}
    H_n=\sum_{a=1}^{N} (\omega + 4 \alpha \beta) b^{\dagger}_a b_a + (\omega \alpha^2 + 2 \alpha \beta \gamma)N \mathbb{I},
    \label{correct basis}
    \end{equation}

provided $\alpha \beta = - \frac{\gamma}{\omega}$, resp. $x=- \frac{1}{2} \arcsin\left(\frac{2 \gamma}{\omega}\right)$. Therefore the rescaled ground state energy is given explicitly (as the whole spectrum of $H_N$) and it satisfies the following 
\begin{align}
e_0 := \omega \alpha^2 + 2 \alpha \beta \gamma=- \frac{2 \gamma^2}{\omega}+ \omega [\sin(\frac{1}{2}\arcsin(\frac{2 \gamma}{\omega}))]^2 \geq -\frac{\gamma^2}{\omega}.
\label{exact energy}
\end{align}
This essentially shows what happens if one uses a 'wrong' basis for a given Hamiltonian with infinitely many degrees of freedom. The original divergent terms $a_a a_a + a^{\dagger}_a a^{\dagger}_a \propto O(n)$ can be eliminated after switching to the correct Fock space generated by $b_a^{\dagger}$ at the cost of producing a divergent multiple of the identity operator as in (\ref{correct basis}).

Another method of getting a lower bound for $e_0$ is based on our considerations from Section \ref{spectral bounds MMM} and \ref{spectral bounds AO}. Let us introduce a pair of composite creation-annihilation operators 
\begin{align}
A&:= a_a a_a \\
[A,A^{\dagger}]&= 2N \mathbb{I}+O(1).
\end{align}

Using the same argument as in Lemma 2, one can show that at large $n$
\begin{align}
a_a^{\dagger}a_a \simeq \frac{1}{N} A^{\dagger} A + \sum_{\lambda_2=0}|\lambda|P_{\lambda}
\end{align}
hence
\begin{align}
H_N \simeq \frac{\omega}{N} A^{\dagger} A + \gamma(A^{\dagger} + A) +  \omega\sum_{\lambda_2=0}|\lambda|P_{\lambda}.
\label{ham trivial}
\end{align}
Let us introduce a new operator $B= A - \epsilon \mathbb{I}\ $, which allows to complete the first two terms of (\ref{ham trivial}) to a full square
\begin{align}
H_N \simeq \frac{\omega}{N} B^{\dagger} B  +  \omega\sum_{\lambda_2=0}|\lambda|P_{\lambda}+ \tilde{e}_0 N \mathbb{I},
\end{align}
with $\tilde{e}_0= 2 \gamma \epsilon + \frac{\omega}{N} \epsilon^2$, provided $\epsilon=-\frac{N \gamma}{\omega}$. Now it follows that $H_N \geq \tilde{e}_0 N \mathbb{I} + O(\frac{1}{n})$ and the true rescaled ground state energy is bounded from below by $\tilde{e}_0$
\begin{equation}
e_0 \geq \tilde{e}_0 = -\frac{\gamma^2}{\omega},
\end{equation}
which agrees with the exact computation (\ref{exact energy}). This shows again that a special treatment of divergent matrix elements of the Hamiltonian ($||A^{\dagger} \psi_{\lambda}|| \propto O(n)$) allows to produce a reasonable lower bound for the divergent part of the ground state energy in the planar limit, i.e. after ignoring terms of order $O(\frac{1}{n})$.

\end{document}